\documentclass[letterpaper, 10 pt, journal]{ieeetran} 

\IEEEoverridecommandlockouts                              % This command is only
\pdfoutput=1
\usepackage{amsmath, amssymb,amsthm}
\usepackage{amsfonts}
\usepackage{graphicx}
\usepackage{enumerate}
\usepackage{ifthen}
\usepackage{multicol}
\usepackage{float}
\usepackage{fancyhdr}
\usepackage[usenames, dvipsnames]{color}
\graphicspath{{figures/}}
\usepackage{mathtools}
\usepackage{multicol}
\usepackage{xcolor}
\usepackage{thmtools,thm-restate} % to restate theorems

\usepackage{dsfont}

\usepackage{cite}
\usepackage{svg}

 \usepackage{cancel} %To be able to cross out text
 \usepackage[normalem]{ulem}

\newtheorem{theorem}{Theorem}
\newtheorem{lemma}[theorem]{Lemma}
\newtheorem{proposition}[theorem]{Proposition}

\newtheorem{rem}{Remark}
\newtheorem{example}{Example}
\newtheorem{ass}{Assumption}
\newtheorem{definition}{Definition}

\newenvironment{proofsketch}{\noindent\textit{Proof sketch.}\ }{\hfill\text{\footnotesize$\blacksquare$}}

\newboolean{showcomments}
\setboolean{showcomments}{true}

\newcommand{\bmat}[1]{\begin{bmatrix} #1\end{bmatrix}}
\newcommand{\Lc}{\mathcal{L}}
\newcommand{\RN}{\mathbb{R}^N}
\newcommand{\R}{\mathbb{R}}
\newcommand{\Nc}{\mathcal{N}}

\newcommand{\ts}[1]{\textsuperscript{#1}}
\newcommand{\1}{\mathds{1}}

\newcommand{\dt}{\frac{\mathrm{d}}{\mathrm{d}t}}
\newcommand{\sat}{\mathrm{sat}}

\newcommand{\dtk}[1]{\frac{\mathrm{d}^{#1}}{\mathrm{d}t^{#1}}}

\newcommand{\vertiii}[1]{{\left\vert\kern-0.25ex\left\vert\kern-0.25ex\left\vert #1 
    \right\vert\kern-0.25ex\right\vert\kern-0.25ex\right\vert}}

\newcommand{\todo}[1]{  \ifthenelse{\boolean{showcomments}}
{\textcolor{ForestGreen}{TO DO:  #1}}{}}
\newcommand{\suggest}[1]{\ifthenelse{\boolean{showcomments}}
{\textcolor{Orange}{(Suggestion: #1)}}{}}
\newcommand{\alain}[1]{\ifthenelse{\boolean{showcomments}}
{\textcolor{Blue}{(Alain says: #1)}}{}}
\newcommand{\jonas}[1]{\ifthenelse{\boolean{showcomments}}
{\textcolor{ForestGreen}{(Jonas says: #1)}}{}}
\newcommand{\kristian}[1]{\ifthenelse{\boolean{showcomments}}
{\textcolor{Blue}{(Kristian says: #1)}}{}}
\newcommand{\emma}[1]{\ifthenelse{\boolean{showcomments}}
{\textcolor{VioletRed}{(Emma says: #1)}}{}}
\newcommand{\ifneeded}[1]{\ifthenelse{\boolean{showcomments}}
{\textcolor{Gray}{#1}}{}}

\newboolean{showedit}
\setboolean{showedit}{true}
\newcommand{\edit}[1]{\ifthenelse{\boolean{showedit}}
{\textcolor{Blue}{#1}}{}}
\newcommand{\etedit}[1]{\ifthenelse{\boolean{showedit}}
{\textcolor{Red}{#1}}{}}
\newcommand{\draft}[1]{\ifthenelse{\boolean{showedit}}
{\textcolor{gray}{#1}}{}}

\newcommand{\mc}{\mathcal}

\usepackage{physics} %to enable use of \norm

\usepackage[font = small]{caption}
\usepackage{subcaption}

\usepackage{tikz}
 \usetikzlibrary{plotmarks}
\usepackage{tikz}
\usepackage{pgfplots}
\pgfplotsset{compat=newest}
\usetikzlibrary{patterns}
\usetikzlibrary{decorations.text}
\usepgfplotslibrary{fillbetween}

\usepackage{lipsum}
\usepackage{mathtools}

\title{\LARGE \bf Compositional design for time-varying and nonlinear coordination
}

\author{{Jonas Hansson and Emma Tegling} 
 \thanks{The authors are with the Department of Automatic Control and the ELLIIT Strategic Research Area at Lund University, Lund, Sweden. Email: \{\tt\small{jonas.hansson, emma.tegling}\}@control.lth.se}
        \thanks{This work was partially funded by Wallenberg AI, Autonomous Systems and Software Program (WASP) funded by the Knut and Alice Wallenberg Foundation. }}

\begin{document}
\maketitle
\begin{abstract}
    This work addresses the design of multi-agent coordination through high-order consensus protocols. While first-order consensus strategies are well-studied---with known robustness to uncertainties such as time delays, time-varying weights, and nonlinearities like saturations---the theoretical guarantees for high-order consensus are comparatively limited.
    We propose a compositional control framework that generates high-order consensus protocols by serially connecting stable first-order consensus operators. Under mild assumptions, we establish that the resulting high-order system inherits stability properties from its components. The proposed design is versatile and supports a wide range of real-world constraints. This is demonstrated through applications inspired by vehicular formation control, including protocols with time-varying weights, bounded time-varying delays, and saturated inputs. We derive theoretical guarantees for these settings using the proposed compositional approach and demonstrate the advantages gained compared to conventional protocols in simulations.
\end{abstract}

\section{Introduction}
Multi-agent coordination is one of the central problems in networked and distributed control. Consensus-seeking in opinions was modeled early on in~\cite{degroot1974reaching} in a discrete-time setting, while \cite{LevineAthans1966, Melzer1971, chu1974decentralized_control, chu1974optimaldecentralized} dealt with the coordination of vehicle strings. The problem was revisited in the early 2000's where a significant research thrust was initiated after some seminal works~\cite{jour:FaxMurray2004,olfati2004consensusdelays,olfati2006flocking, spanos2005Dynamicconsens, moreau2005TACstabtimedependent, jour:Jadbabaie2003,ren2005consensuschangingtopology}. These established many of the fundamental properties of second- and first-order consensus protocols. Based on these works, we know that the first-order consensus protocol $\dot{x}=-L(t)x$, where $L(\cdot)$ is a time-varying graph Laplacian that encodes relative feedback, is robust to delays and time-varying topology. Furthermore, that consensus has a wide range of applications, ranging from swarming robots, vehicle platoons, frequency synchronization, and describing natural flocking behaviors \cite{olfati2007consensuscoop}.

In this work, we will consider coordination among higher-order agents. If the first-order consensus protocol is $\dot{x}=-Lx$, the second-order is $\ddot{x}=-L_\mathrm{vel}\dot{x}-L_\mathrm{pos}x$, then there is a natural generalization to a general $n\ts{th}$-order consensus protocol, which is
$x^{(n)}=-L_\mathrm{(n-1)}x^{(n-1)}-\dots-L_{(0)}x.$
Here, the goal is to coordinate in position, velocity, and the high-order derivatives. This high-order consensus problem was first considered in \cite{ren2006highorder,jour:Ren2007}, where some sufficient conditions for system stability were also established.

Our motivation for revisiting this problem comes from studies of the dynamic behaviors of, in particular, large-scale coordinating multi-agent networks. Vehicular platoons suffering from error propagation and string stability challenges~\cite{jour:Swaroop2006, jour:studli2017StringConcepts}, fundamental limitations in terms of formation coherence~\cite{jour:Bamieh2008, jour:Bamieh2012Sep}, and lately scale fragilities~\cite{jour:TEGLING2023, jour:TEGLINGMiddleton2019} in second- and higher-order consensus. Here, poor dynamical behaviors---or even instability---can emerge when networks grow large in a manner that is hard to foresee from the original, distributed, control design. This calls for methods to construct coordination protocols that allow for a modular and scalable network design. An early and influential such approach based on passivity was~\cite{arcak2007passivitycoord}; here, we take an alternative route. 

Apart from challenges related to large-scale dynamic behaviors, even the problem of stabilization is non-trivial in higher-order coordination, and more difficult than in first- and second-order protocols. For instance, the first-order protocol permits time-varying topologies with certain time-delayed measurements \cite{jour:Lu2017consensusdelays}, saturations \cite{jour:LiConsensusSaturation}, and directed topologies provided the network is sufficiently connected over time. For the second-order linear and time-invariant consensus protocol, stability can be guaranteed, provided that $L_\mathrm{vel}=r_\mathrm{vel}L$ and $L_\mathrm{pos}=r_\mathrm{pos}L$, where $L$ is symmetric and contains a spanning tree. When the symmetry condition is broken, as in the case of a directed cycle graph, then the second-order linear consensus protocol can become unstable \cite{jour:Studli2017}. There are also many works that have derived sufficient conditions for this protocol when subject to time-delays \cite{Gao2023secondorderdelay}, time-varying structures \cite{li2022secondvaryinganddelays}, and various nonlinearities \cite{lyu2016consensusinputsaturation}. However, these results most often depend on global or absolute knowledge of the positions and velocities. Higher-order protocols have been shown to lack scalable stability in sparse networks~\cite{jour:TEGLING2023}, meaning that a loss of closed-loop stability is inevitable without a {(re-)}~tuning based on global knowledge.  Other works on high-order coordination include~\cite{trinidade2014highorderlqr} that studied LQR, \cite{li2024secondvaryinganddelays, TIAN201212highorderdelays} time-varying topology and delays, and \cite{Rezaee2015highorderforaverageconsensus,he2011, Liu2010hinfcontrolabsfeedback}, where consensus is achieved, but with the help of absolute feedback.

In this work, we propose a novel consensus protocol for achieving coordination in a network of $n\ts{th}$-order integrators. Our proposed control design is based on the idea of first designing the closed-loop system and then identifying the corresponding control law. The class of desired closed-loop systems can be written as the composition of $n$ simple first-order consensus systems, that is,
$$\left(\dtk{} + \Lc_n\right)\circ \cdots \circ\left(\dtk{} +\Lc_1\right)(x)=0,$$ where each $\Lc$ describes a, potentially nonlinear and time-varying operator that generalizes the graph Laplacian in the LTI case. 
Due to its compositional structure, we will call this the \emph{compositional consensus} system. In the second-order case, this can be expanded to
$$\ddot{x}=u(x,t)=-\Lc_2(\dot{x}+\Lc_1(x,t),t)-\dtk{}\Lc_1(x,t).$$
Under relatively mild conditions on the operators $\Lc_k$, essentially that their corresponding first-order protocols achieve consensus, we can show that this control design guarantees that the solution $x$, $\dot{x}$ and all its first $n-1$ derivatives will coordinate and reach an $n\ts{th}$-order consensus. We note that this is independent of the underlying graph structure.  

To illustrate the strength of the compositional consensus, we also demonstrate how to apply this controller when the composing components $\Lc_k$ correspond to saturated, time-varying, and time-delayed consensus protocols, building on and partially extending results existing in literature. In particular, we prove a general result on the stability of consensus under saturated control inputs. We formally and through case studies show that compositional consensus has superior stability and performance than a more na\"ive higher-order protocol. The implementation of the protocol remains localized, but requires some additional signaling in an $n-$hop neighborhood, or message-passing between nearest neighbors. Due to the strong robustness towards time-delays and time-varying connectivity, such signaling need not be ideally implemented. 

% \emma{Later, we can revise this intro and connect to more literature on composition/interconnection eg paper by Lestas and vinnicombe. }

\paragraph*{Paper Outline}
The remainder of the paper is organized as follows. We next introduce some notation and preliminaries, followed by an introduction of compositional consensus in~Section~\ref{sec:problem}. Section~\ref{sec:main} presents our main result along with key lemmas used in the proof. Section~\ref{sec:case} studies some selected first-order consensus protocols that can be used in the compositional design. In Section~\ref{sec:applications} we illustrate our result through numerical simulations, and Section~\ref{sec:conclusions} concludes the paper.

\subsection{Mathematical preliminaries}
\paragraph*{Graph theory}
We represent a directed graph as \( \mathcal{G} = (\mathcal{V}, \mathcal{E}) \), where \( \mathcal{V} = \{1, \dots, N\} \) is the set of nodes, and \( \mathcal{E} \subset \mathcal{V} \times \mathcal{V} \) is the set of edges. The graph is associated with a weighted adjacency matrix \( W \in \mathbb{R}^{N \times N} \), where \( W_{i,j} > 0 \) if and only if \( (j, i) \in \mathcal{E} \), i.e., agent \( j \) influences agent \( i \).
The corresponding graph Laplacian is defined as
\[
L = D - W,
\]
where \( D \) is the diagonal degree matrix with \( D_{i,i} = \sum_{j=1}^N W_{i,j} \). The graph Laplacian \( L \) has zero row sum and encodes the relative feedback structure of the network.

A graph is said to contain a \emph{directed spanning tree} if there exists a node \( k \in \mathcal{V} \) such that all other nodes \( j \in \mathcal{V} \setminus \{k\} \) are reachable via a directed path from \( k \). If this condition holds, the Laplacian \( L \) has a simple zero eigenvalue, and all other eigenvalues have strictly positive real parts.

We also make use of \emph{\( \delta \)-graphs} as defined in~\cite{Moreau2004stabconsensus}. Given a threshold \( \delta > 0 \), the \( \delta \)-graph associated with \( W \) is a subgraph where an edge \( (j, i) \) is retained if and only if \( W_{i,j} \geq \delta \).

\paragraph*{Norms and other notation}
We denote vector and matrix norms using \( \|\cdot\| \). For vectors \( x \in \mathbb{C}^N \), we use the \( \infty \)-norm
\[
\|x\|_\infty = \max_i |x_i|,
\]
and for matrices \( C \in \mathbb{C}^{N \times N} \), the induced matrix norm
\[
\|C\|_\infty = \max_i \sum_j |C_{ij}|.
\]
Seminorms are denoted \( \vertiii{\cdot} \), following the notation in~\cite{bullo2024-CTDS}. These are functions satisfying the triangle inequality \(\vertiii{x_1+x_2}\leq \vertiii{x_1}+\vertiii{x_2}\) and absolute homogeneity $\vertiii{a x}\leq |a|\vertiii{x}$. When the context is clear, we drop the explicit time-dependence \( x(t) \) in the notation. We write \( \dtk{j} x = x^{(j)} \) for the \( j \)\ts{th} time derivative, and use \( \partial_t F(x,t) \) for partial derivatives.

Function composition is written \( (f \circ g)(x) = f(g(x)) \). When composing time-varying functions, we use the convention
\[
(\Lc_2 \circ \Lc_1)(x) := \Lc_2(\Lc_1(x,t), t).
\]
A continuous function $\gamma(\cdot)$ is said to be of class~$\mathcal{K}$ if it is non-negative and strictly increasing. A continuous function $\beta(\cdot,\cdot)$ belongs to class $\mathcal{K}\mathcal{L}$ if it for any fixed $s$, $\beta(\cdot,s)$ is of class~$\mathcal{K}$, and for any fixed $r$, $\beta(r,\cdot)$ is decreasing with respect to $s$ and satisfy $\lim_{s\to\infty}\beta(r,s)=0$. This follows the standard notation of \cite{jour:khalil2002nonlinear}.

\subsection{Consensus}
Due to the presumed lack of absolute feedback (see Assumption~\ref{ass:kisinvariant}), the relevant notion of stability in this work is instead one of consensus among the agents. It is defined as follows. 
\begin{definition}[Consensus]\label{def:consensus}
Let \( x(t) \in \mathbb{R}^N \) be the state of a multi-agent system governed by
$\dot{x} = f(x,t).$
The system is said to achieve \emph{consensus} if
\[
\lim_{t \to \infty} |x_i(t) - x_j(t)| = 0, \quad \text{for all } i \neq j.
\]
\end{definition} \noindent 
It is well known that the simple, linear, continuous-time consensus protocol 
\[
\dot{x} = -Lx,
\]
achieves consensus, where \( L \) is a graph Laplacian, provided the underlying graph contains a directed spanning tree \cite{jour:Zhiyun2005directedconsensus,jour:Ren2005digraphproof}.
In high-order coordination problems, synchronizing the positions and higher-order derivatives, such as velocities and accelerations, is often desirable. This motivates the following generalization (see also \cite{ren2006highorder}):

\begin{definition}[\( n\ts{th} \)-Order consensus]
Let \( x(t) \in \mathbb{R}^N \) evolve according to
\[
\frac{\mathrm{d}^n x}{\mathrm{d}t^n} = f(x,t).
\]
The solution $x$ is said to achieve \( n\ts{th} \)-order consensus if
\[
\lim_{t \to \infty} |x_i^{(k)}(t) - x_j^{(k)}(t)| = 0,\, \forall i \neq j \text{ and } k = 0, \dots, n-1.
\]
\end{definition}
\noindent This definition captures the idea that all agents eventually align in their positions and higher-order dynamics, like velocities and accelerations.

\section{Problem Setup}\label{sec:problem}
In this work we consider a network consisting of $N$ identical agents with $n\ts{th}$-order integrator dynamics, that is,
\begin{equation}
\label{eq:nthorderintegrator}
    x^{(n)}(t)=u(x,t),
\end{equation}
$x(0) = x_0, ~\dot{x}(0) = \dot{x}_0,\ldots,~x^{(n)}(0) = x^{(n)}_0$. Our proposed control design can be compared with a Youla-Kucera parametrization, where we first design the closed-loop system to be
\begin{equation}
   \left(\dt+ \Lc_n\right)\circ \cdots\circ \left(\dt+ \Lc_2\right)\circ \left(\dt+ \Lc_1\right) \!(x)\!=\!0,%\ur(t),
    \label{eq:nonlinear_serial}
\end{equation}
where each operator $\Lc_k(\cdot,t): \RN \to \RN$ is allowed to be time-dependent and potentially nonlinear. Our controller is then chosen to be the one that achieves this closed-loop system, that is,
\begin{equation}\label{eq:controller_nthorder}
    u(x,t)=x^{(n)}(t)-  \left(\dt+ \Lc_n\right)\circ \cdots\circ \left(\dt+ \Lc_1\right) \!(x).
\end{equation}
The closed-loop design matches the behavior of $n$ dynamical systems in a series interconnection. The system can be analyzed in the following simple state-space form
\begin{equation}\label{eq:statespace}
    \bmat{\dot{\xi}_1\\ \vdots \\\dot{\xi}_n}=\bmat{ -\Lc_1(\xi_1,t)+\xi_2\\ \vdots \\ -\Lc_n(\xi_n,t)}.
\end{equation}
    where $\xi_1=x$ and $\xi_{k+1}=\dot{\xi}_{k}+\Lc_k(\xi_k,t)$. We want to emphasize that the state-space formulation~\eqref{eq:statespace} is the key to the scalability of the compositional consensus formulation. The series interconnection of dynamical systems has some favorable properties. For instance, any series interconnection of strongly contracting systems will itself be strongly infinitesimally contracting \cite[Theorem~3.23]{bullo2024-CTDS}. We also want to highlight the connection to the literature on distributed optimization, using gradient tracking \cite{CarnevaleGradTrack, Dhullipalla2024GradTrack}, and also used for dynamic average consensus \cite{kia2019tutorialdynamic, AldanaLopez2022DynamicConsensus}. One key difference between the mentioned work and ours is that we will mainly focus on the scenario where $\Lc_k$ are restricted to using only \textit{relative feedback}. A limitation to relative feedback poses severe challenges in coordination control design, see e.g.~\cite{jour:Bamieh2012Sep}, and is motivated by a fundamental difficulty in many applications to capture absolute position, phase, etc., whereas the corresponding relative measurement is readily available. The following Assumption captures this limitation.
\begin{ass}[Relative feedback]\label{ass:kisinvariant} 
The feedback operators satisfy  $\Lc_k(z(t)+\1 a(t),t)=\Lc_k(z(t),t)$ for any $z$, $a(\cdot)$, and~$t$. 
\end{ass}
Under suitable and relatively simple conditions, it is possible to show that the solution~$x(t)$ of the closed-loop system~\eqref{eq:nonlinear_serial} will converge to an $n\ts{th}$-order consensus. Furthermore, that sparsity of the individual operators $\Lc_k$ implies sparsity of the controller $u(x,t)$, as defined by \eqref{eq:controller_nthorder}. 

\begin{example}\label{example:serial_consensus}
%\emma{Serial consensus (LTI). With picture?}
In the linear, time-invariant case, the composition~\eqref{eq:nonlinear_serial} may capture the \emph{serial consensus} protocol. Here, the closed loop matches regular consensus protocols connected in a series. 
%This work builds on the idea of serial consensus where the closed-loop system is designed to match the behavior of regular consensus protocols put in a series interconnection. 
The $n\ts{th}$-order serial consensus system can, in the Laplace domain, be represented as 
$$(sI+L_n)\cdots(sI+L_2)(sI+L_1)X(s)=U_\mathrm{ref}(s).$$
One property that makes this system interesting is the simple condition for stability. That is, if each of the graphs underlying the $L_k$'s contain a directed spanning tree, then this high-order consensus protocol will achieve an $n\ts{th}$-order consensus, assuming a decaying input signal $\|u_\mathrm{ref}(t)\|\to 0$ \cite{jour:hansson2023scalable}. The serial consensus protocol can also be used to construct linear time-invariant consensus protocols for vehicular formations with a strong notion of scalable performance \cite{jour:hansson2025tcns ,jour:hansson2024lcss}. It, therefore, avoids issues with scale fragility~\cite{jour:TEGLING2023} and string instability~\cite{jour:studli2017StringConcepts} affecting conventional consensus protocols. 

To implement serial consensus, additional signaling may be needed in the multi-agent system. This can be seen in the second-order serial consensus where the control law is
$$u(x,t)= -(L_1+L_2)\dot{x}-L_2L_1x.$$
Here, the velocity feedback can be realized immediately through local measurements. For the second term, each agent can first aid in calculating $e=L_1x$, then message pass this measurement to their followers so that they can compute the relative differences~$L_2 e$. In general, it is possible to compute the local control law for the $n\ts{th}$-order serial consensus through the use of $n-1$ local message passes; the local consensus protocol is only dependent on relative measurements within an $n$-hop neighborhood (at most) of each agent. 
\end{example}

\subsection{Assumptions}
\label{sec:assumptions}
To prove the main result, we will use the following assumptions. The first ones cover the operators $\Lc_k$ for $k\leq n-1$. 
First, we impose a standard technical assumption used to establish the existence and uniqueness of a solution.
\begin{ass}\label{ass:kislipschitz}
    The $\Lc_k(z,t)$ is Lipschitz in $z$ with a global Lipschitz constant independent of $t$ and are, for any fixed $z$, piecewise continuous in $t$.  
\end{ass}
The next assumption is one of input-to-state stability (ISS) for the individual subsystems in the composition.    
\begin{ass}\label{ass:kislocallyISS}
    If $\|w(t)\|\leq M_k$ for all $t\geq T_0$, then the system $\dot{z}=\Lc_k(z,t)+w(t)$ is ISS with respect to some seminorm $\vertiii{\cdot}$, that is:
    $$\vertiii{z(t)}\leq \beta_k(\vertiii{z(T_0)},t)+\gamma_k(\sup_{t\geq T_0}\|w(t)\|)$$
    where $\beta_k\in \mathcal{K}\mathcal{L}$, $\gamma_k\in \mathcal{K}$, and the seminorm satisfies $\vertiii{z}=0 \iff z\in \mathrm{span}(\1)$.
\end{ass}
\noindent It implies consensus of the individual subsystems and will be needed to prove consensus of the composed system.  Finally,

\begin{ass}\label{ass:kiscont}
     Let $\Lc_k\in C^{n-1-k}$ be chosen such that $\|\dtk{j} \Lc_k(z,t)\|\leq \alpha_{k,j}(\max_{0\leq i\leq j}\|z^{(j)}\|)$ for some functions $\alpha_{k,j}\in \mathcal{K}$, for all $j\leq n-k-1$, and all time $t\geq 0$.
\end{ass}
\noindent This assumption asserts a smoothness of the composing operators~$\Lc_k$. With this assumption, we can prove that coordination of~$\xi_k$ is equivalent to the coordination of~$x$, $\dot{x},\dots$, $x^{n-1}$. With these assumptions established, we are ready to state our main theorems.

\pagebreak

\section{Main Results}\label{sec:main}
Consider the following result, which establishes that the composition of consensus protocols according to~\eqref{eq:nonlinear_serial} will also achieve consensus.
\begin{theorem}\label{thrm:asymptoticconsensus}
Let each subsystem \( \mathcal{L}_k \), implement relative feedback 
according to Assumption~\ref{ass:kisinvariant}, and be chosen such that each
unperturbed system
\[
\dot{z}_k = \mathcal{L}_k(z_k, t)
\]
admits a unique solution for every initial condition \( z_k(0) \) that %, and 
satisfies
\[
\lim_{t \to \infty} \|z_k(t) - \mathbf{1} a_k(t)\| = 0,
\]
for some function $a_k$. Assume additionally that each \( \mathcal{L}_k \) satisfies Assumptions~\ref{ass:kislipschitz}--\ref{ass:kiscont} for \( k = 1, \dots, n-1 \). Then, the compositional consensus system~\eqref{eq:nonlinear_serial} admits a unique solution~\( x \), and this solution achieves \( n\)\ts{th}-order consensus.
\end{theorem}

We now present the lemmas that form the basis of the proof of Theorem~\ref{thrm:asymptoticconsensus}.
\begin{lemma}\label{lem:consensusequivalnce}
    If all $\Lc_k$, $k=1,\ldots,n$ implement relative feedback according to~Assumption~\ref{ass:kisinvariant}, and satisfies the smoothness and boundedness Assumption~\ref{ass:kiscont} for all $ k\leq n-1$, and all times $t\geq 0$. Then, the following two are equivalent:
\begin{itemize}
    \item[i)] The solution $x$ of the compositional consensus \eqref{eq:nonlinear_serial} achieves $n\ts{th}$-order consensus for any initial condition;
    \item[ii)] the states $\xi_k$, $k=1,\ldots,n$, of \eqref{eq:statespace} achieve first-order consensus for any initial condition.
\end{itemize}
\end{lemma}\vspace{0.5ex}

\begin{proofsketch} (Full proof given in Appendix~\ref{app:conensusequivalenceproof}.)
We prove the equivalence by induction in two directions. First, we show that the initial condition of \( x \) and its first \( n-1 \) derivatives uniquely determine the initial conditions of the states \( \xi_k \). Then, using a similar argument, the reverse direction can be proven.

Since \( x(t) = \xi_1(t) \), and Assumption~\ref{ass:kiscont} ensures sufficient smoothness, we may recursively differentiate this relation to recover all \( \xi_k(t) \). The supporting Lemma~\ref{lem:consensus_invariance}, a consequence of the relative feedback Assumption~\ref{ass:kisinvariant}, allows us to bound the terms \( \dtk{j} \mathcal{L}_k(\xi_k, t) \) in terms of deviations from consensus
\[
\|\dtk{j} \Lc_k(\xi_k, t)\| \leq \alpha_{k,j}(\max_{0 \leq i\leq j}(\|\xi_k^{(i)}-\1 b_{k+i}(t)\|).
\]
We then apply induction in \( k \) to show that \( n \)\ts{th}-order consensus of \( x \) implies first-order consensus of all \( \xi_k \) and induction in \( j \) to show the converse. The full details are provided in Appendix~\ref{app:conensusequivalenceproof}.
\end{proofsketch}
\vspace{2mm}

Having established the equivalence between \eqref{eq:nonlinear_serial} and \eqref{eq:statespace}, we now show that the latter achieves consensus under relatively mild conditions.

\begin{lemma}\label{lem:statespaceconsensus}
Let each subsystem \( \mathcal{L}_k \) implement relative feedback according to~Assumption~\ref{ass:kisinvariant}, and assume that the unperturbed system
\[
\dot{z}_k = \mathcal{L}_k(z_k, t)
\]
admits a unique solution for any initial condition \( z_k(0) \), and satisfies
\[
\lim_{t \to \infty} \|z_k(t) - \mathbf{1} b_k(t)\| = 0
\]
for some function \(b_k \). Assume additionally that each \( \mathcal{L}_k \) satisfies Assumptions~\ref{ass:kislipschitz} and~\ref{ass:kislocallyISS} for \( k = 1, \dots, n-1 \). Then, the states \( \xi_k \) in~\eqref{eq:statespace} admit a unique solution, and satisfy
\[
\lim_{t \to \infty} \|\xi_k(t) - \mathbf{1} a_k(t)\| = 0
\]
for some functions \( a_k \).
\end{lemma}

\begin{proofsketch}(Full proof given in Appendix~\ref{app:lemstatespaceconsensus}.)
Existence and uniqueness follow from Carathéodory's existence and uniqueness theorem, which here follows from %is implied by 
Assumption~\ref{ass:kislipschitz}. The solution for \( \xi_n \) exists and achieves consensus by assumption. The remaining states \( \xi_k \) are shown to reach consensus through induction on \( k \).

In particular, through Assumption~\ref{ass:kisinvariant} we establish that
\[
\dot{\tilde{\xi}}_k = -\Lc_k(\tilde{\xi}_k, t) + w(t),
\]
where \( \tilde{\xi}_k(t) := \xi_k(t) - \1 \int_0^t a_{k+1}(\tau)\, \mathrm{d}\tau \), and \( w(t) := \xi_{k+1}(t) - \1 a_{k+1}(t) \). The input \( w(t) \) converges to zero due to the inductive hypothesis.

Then, by the local ISS property in Assumption~\ref{ass:kislocallyISS}, it follows that \( \tilde{\xi}_k \) achieves first-order consensus, which in turn implies consensus of \( \xi_k \). Repeating this argument recursively establishes the result for all \( \xi_k \). 
\end{proofsketch}

\vspace{2mm}
With Lemmas~\ref{lem:consensusequivalnce} and~\ref{lem:statespaceconsensus}, the main result in Theorem~\ref{thrm:asymptoticconsensus} is now readily established. Together, these lemmas provide sufficient conditions for the solution \( x \) to achieve \( n \)th-order consensus. Theorem~\ref{thrm:asymptoticconsensus} is thereby proven.
\begin{rem}
The existence and uniqueness of the solution \( x \) should be interpreted as a weak solution, i.e., a function that satisfies the differential equation~\eqref{eq:nonlinear_serial} almost everywhere. If a smooth solution is desired---i.e., one that satisfies the equation pointwise---one may strengthen Assumption~\ref{ass:kiscont} by requiring \( \Lc_k \in C^{n-k} \) instead of \( \Lc_k \in C^{n-k-1} \).
\end{rem}

\subsection{Example: serial consensus}
To better illustrate Theorem~\ref{thrm:asymptoticconsensus}, consider again the simplest case where each operator is a linear time-invariant function, i.e., \( \Lc_k(x,t) = L_k x \), which leads to the compositional consensus system also known as \emph{serial consensus}:
\begin{equation}\label{eq:linear_serial}
    \left(\dtk{} + L_n\right) \cdots \left(\dtk{} + L_1\right)(x) = 0.
\end{equation}
We now show the following.

\begin{proposition}\label{prop:serial}
    If each graph Laplacian \( L_k \) in~\eqref{eq:linear_serial} contains a (possibly different) directed spanning tree, then \( x \) achieves \( n \)\ts{th}-order consensus for any initial condition.
\end{proposition}\vspace{0.5ex}

The proof of this proposition serves as an example of how to apply our compositional consensus result.

\setcounter{example}{0}
\begin{example}[\textbf{continued}]
We verify the assumptions of Theorem~\ref{thrm:asymptoticconsensus} for the case \( \Lc_k(x,t) = L_k x(t) \). First, since
\[
\|L_k x - L_k y\| \leq \|L_k\| \|x - y\|,
\]
each operator is Lipschitz and time-invariant, so Assumption~\ref{ass:kislipschitz} is satisfied. The invariance property of Assumption~\ref{ass:kisinvariant} follows from the definition of the Laplacian, since \( L_k(\1 a(t)) = 0 \) for any scalar function \( a(t) \).

The most involved step is verifying Assumption~\ref{ass:kislocallyISS}. We define the seminorm \( \vertiii{z}_k := \|L_k z\| \), which is valid since \( L_k \) has a simple zero eigenvalue with corresponding eigenvector~\( \1 \) (by the spanning tree condition). Consider the perturbed consensus system
\[
\dot{z}(t) = -L_k z(t) + w(t).
\]
Its general solution is given by
\[
z(t) = e^{-L_k(t - T_0)} z(T_0) + \int_{T_0}^t e^{-L_k(t - \tau)} w(\tau)\, \mathrm{d}\tau.
\]
Premultiplying by \( L_k \), taking norms, and using the bound \( \|L_k e^{-L_k t}\| \leq M_k e^{-\alpha t} \), valid for some \( M_k, \alpha > 0 \), yields
\[
\|L_k z(t)\| \leq M_k e^{-\alpha(t - T_0)} \|L_k^+\| \|L_k z(T_0)\| + \frac{M_k}{\alpha} \sup_{\tau \geq T_0} \|w(\tau)\|,
\]
where \( L_k^+ \) denotes a pseudoinverse of $L_k$. From this inequality, one can identify \( \beta_k(\cdot) \in \mathcal{K}\mathcal{L} \) and \( \gamma_k(\cdot) \in \mathcal{K} \), verifying the local ISS property.
That the system \( \dot{z} = -L_n z \) asymptotically reaches consensus is well known, but it is also a direct consequence of the above discussion.

Finally, Assumption~\ref{ass:kiscont} concerns the smoothness of \( \Lc_k \). Since \( \partial_t \Lc_k(x,t) = 0 \), \( \partial_x \Lc_k(x,t) = L_k \), and all higher-order derivatives
\[
\frac{\partial^{i+j}}{\partial x^j \partial t^i} \Lc_k(x,t) = 0
\]
vanish, it follows that \( \Lc_k \in C^{n - 1 - k} \) as required.
The time derivative of \( \Lc_k(x,t) \) is
\[
\dtk{j} \Lc_k(x,t) = L_k x^{(j)},
\]
so
\[
\left\| \dtk{j} \Lc_k(x,t) \right\| \leq \|L_k\| \|x^{(j)}\|,
\]
which satisfies Assumption~\ref{ass:kiscont} with bounding functions \( \alpha_{k,j}(r)=\|L_k\||r| \), which clearly are of class $\mathcal{K}$. 
\end{example}

This result confirms the stability of the linear serial consensus system previously established by different methods in~\cite{jour:hansson2023scalable}. Although many assumptions need to be checked, most are straightforward. The more involved ones---Assumptions~\ref{ass:kislocallyISS} and the stability of unperturbed first-order system---can be verified using classical first-order consensus theory, as we will demonstrate in the following applications.

\subsection{Implementation of compositional consensus}
The implementation of compositional consensus raises a few key questions: 1) Is the protocol implementable using only local and relative feedback? 2) Will the control signal be well-defined? 

The answer to the first question is yes; provided that each $\Lc_k$ implements relative local feedback, the resulting feedback will also be relative and local. To illustrate this, consider the third-order case
\begin{multline*}    
    x^{(3)}=-\Lc_3(\ddot{x}+\Lc_2(\dot{x}+\Lc_1(x,t),t)+\dtk{}\Lc_1(x,t),t)\\-\dtk{}\Lc_2(\dot{x}+\Lc_1(x,t),t) -\dtk{2}\Lc_1(x,t) 
\end{multline*}
and suppose that the unweighted adjacency matrix $W_k$ encodes the communication structure of $\Lc_k$, $k = 1,2,3$. That is, $[W_k]_{i,j}=1 \iff [\Lc_k(z+\mathbf{e}_j,t)-\Lc_k(z,t)]_i\neq 0$ where $\mathbf{e}_j$ is the $j\ts{th}$ unit vector.  
We may now work backward to deduce the adjacency matrix encoding the full feedback. The term $z_1=\Lc_1(x,t)$ depends on measurements coming from the graph associated with $W_1$ and so will all its higher derivatives. Let $z_2=\Lc_2(\dot{x}+z_1,t)$, which then depends on signals encoded by $W_2(W_1+I)$. Finally, since $z_3=\Lc_3(\ddot{x}+z_2+\dot{z}_1)$, these signals will be encoded by $W_3(I+W_2(I+W_1)+W_1)$. In general, we see that all measurements needed in the feedback are contained by the graph associated with $ (W_k+I)(W_{k-1}+I)\cdots (W_1+I)$. In the special case where the $W_k$ are identical, the product implies a $k$-hop neighborhood in the graph in question.

As to the second question, it is not in general guaranteed that the highest derivative $x^{n}(t)$ is well-defined for all $t$ in the fairly general setting of~Theorem~\ref{thrm:asymptoticconsensus}. 
%In the setting of Theorem~\ref{thrm:asymptoticconsensus}, a well-defined $x^{n}$ is not generally guaranteed to exist for all $t$. 
The issue can be effectively illustrated by considering the second-order case. 
\[u(x,t)=-\Lc(\dot{x}+\Lc_1(x,t),t) -\dtk{}\Lc_1(x,t).\]
While the first term suffers no problem, the term  $\dtk{}\Lc_1(x,t)$ may be problematic, since Theorem~\ref{thrm:asymptoticconsensus} only requires $\Lc_1 \in C^0$.
The derivative term may, therefore, instead be interpreted in terms of the Dini derivative, that is, 
\[
D^+(\Lc_1(x(t),t)=\limsup_{\Delta t>0,\Delta t\to 0}\dfrac{\Lc_1(x(t+\Delta t), t+\Delta t)}{\Delta t}.
\]
This will always be well-defined; see e.g. \cite[A.7]{bullo2024-CTDS}, however, potentially unbounded. From a more practical view, one can consider a function that approximates the derivative almost everywhere. This is relevant when using nonlinear functions like the $1$- and $\infty$-norms, saturations, and dead zones. In the case of saturations one may use $D^+\sat(x_i(t))\overset{\mathrm{a.e.}}{=} \dot{x}_i(t)\mathbf{I}_{(-1,1)}(x_i(t))$
where $\mathbf{I}_{(-1,1)}(\cdot)$ is an indicator function.

\section{Applications of Compositional Consensus}\label{sec:case}
In this section, we explore applications of Theorem~\ref{thrm:asymptoticconsensus} in representative nonlinear and time-varying networked systems. This amounts to verifying whether the protocols satisfy the Assumptions in Section~\ref{sec:assumptions}.

\subsection{Saturated consensus}
 A common type of nonlinearity in control systems is \emph{saturation}, which arises due to actuator limitations or other physical constraints. This example shows that saturated signals can be handled within the compositional consensus framework.

Earlier works, such as~\cite{jour:LiConsensusSaturation}, have established asymptotic consensus stability for the unforced system
\[
\dot{z} = -\mathrm{sat}(Lz).
\]
This makes the corresponding operator %\( \Lc_n(x,t) = Lx \) 
admissible as the outermost function in the compositional consensus~\eqref{eq:nonlinear_serial}, i.e., $\Lc_n(z,t) = \mathrm{sat}(Lz)$. 
In the following, we extend the analysis to bounded-input scenarios, thereby enabling the use of saturation-based dynamics also for \( \Lc_{n-1}(z,t) \); this follows since the remaining assumptions are simple to check. The following proposition proves the applicability to Assumption~\ref{ass:kislocallyISS}.

\begin{proposition}\label{prop:saturation}
Consider the consensus system
\[
\dot{z} = -\mathrm{sat}(Lz) + d(t),
\]
where \( L \) is a graph Laplacian that contains a directed spanning tree. Then there exists a constant \( d_{\max} > 0 \) such that, for all disturbances satisfying \( \|d(t)\|_\infty < d_{\max} \), the disagreement satisfies
\[
\|Lz(t)\|_\infty \leq \beta(\|Lz(0)\|_\infty, t) + \gamma(\sup_{t \geq 0} \|d(t)\|_\infty)
\]
for some functions \( \beta \in \mathcal{KL} \) and \( \gamma \in \mathcal{K} \). 
\end{proposition}
\noindent The proof is provided in Appendix~\ref{app:proof_saturation}. This result is applied to vehicular coordination in Section~\ref{sec:satcoord}.

\subsection{Time-varying linear dynamics}
Another class of systems that can be used within the compositional consensus framework is linear time-varying dynamics. The following proposition is a straightforward application of~\cite[Theorem~1]{Moreau2004stabconsensus}.%extension of

\begin{proposition}\label{prop:timevarying}
Consider the first-order consensus system
\[
\dot{z} = -L(t)z + d(t),
\]
where \( L(t) \) is a piecewise continuous Metzler matrix. Let \( A(t) = \int_t^{t+T} L(\tau)\, \mathrm{d}\tau \). Suppose there exists \( \delta > 0 \), a fixed node \( k \), and a time \( T > 0 \) such that, for every \( t \), the \( \delta \)-digraph associated with \( A(t) \) contains a node \( k \) that is reachable from all other nodes.
Define the disagreement seminorm
\[
\vertiii{z} := \left\| \left(I - \frac{\1 \1^\top}{N} \right) z \right\|.
\]
Then, the solution satisfies
\[
\vertiii{z(t)} \leq \beta\left( \vertiii{z(0)}, t \right) + \gamma( \textstyle{\sup_{t \geq 0} }\|d(t)\| ),
\]
for some functions \( \beta \in \mathcal{KL} \) and \( \gamma \in \mathcal{K} \).
\end{proposition}

\begin{proof}
The cited theorem \cite[Theorem~1]{Moreau2004stabconsensus} establishes that the consensus equilibrium of the unperturbed system (i.e., \( d(t) = 0 \)) is uniformly exponentially stable, meaning
\[
\vertiii{z(t)} \leq M e^{-\alpha t} \vertiii{z(0)}
\]
for some constants \( M > 0 \) and \( \alpha > 0 \). Let \( \Phi(t,t_0) \) denote the state transition operator of the unforced system \( \dot{z} = -L(t)z \). Then, by the variation of constants formula (see, e.g.,~\cite[Chapter~2]{hinrichsen2005mathematical}), the solution to the forced system is
\[
z(t) = \Phi(t,t_0)z(t_0) + \int_{t_0}^t \Phi(t,\tau) d(\tau)\, \mathrm{d}\tau.
\]

Applying the projection operator \( I - \frac{\1 \1^\top}{N} \) to both sides, and using the fact that \( \Phi(t,t_0) \) preserves the consensus subspace, we get
\[
\vertiii{z(t)} \leq M e^{-\alpha (t - t_0)} \vertiii{z(t_0)} + \int_{t_0}^t M e^{-\alpha(t - \tau)} \|d(\tau)\|\, \mathrm{d}\tau.
\]
Using the standard exponential convolution estimate, we obtain
\[
\vertiii{z(t)} \leq M e^{-\alpha (t - t_0)} \vertiii{z(t_0)} + \frac{M}{\alpha} \sup_{t \geq 0} \|d(t)\|,
\]
which is an ISS-type bound of the desired form.
\end{proof}

For \( \Lc_k(z,t) = L_k(t)z \), it is straightforward to verify Assumptions~\ref{ass:kislipschitz}--\ref{ass:kislocallyISS}. What remains is to establish the smoothness condition in Assumption~\ref{ass:kiscont}. Note that
\begin{align*}
\partial_{t^j} (L_k(t)z) &= L_k^{(j)}(t)z,\\
\partial_{t^j} \partial_z (L_k(t)z) &= L_k^{(j)}(t), \\
\partial_{t^j} \partial_z^2 (L_k(t)z) &= 0.
\end{align*}
This shows that \( \Lc_k \in C^n \) if and only if \( L_k(t) \) is \( n \)-times continuously differentiable.

Furthermore, applying the product rule yields
\[
\dtk{j}(L_k(t)z(t)) = \sum_{i=0}^j \binom{j}{i} L_k^{(i)}(t)z^{(j-i)},
\]
which can be uniformly bounded in terms of \( \max_{0\leq i\leq j}\left\{ \|z^{(i)}\| \right\} \) provided that \( \|L_k^{(i)}(t)\| \leq M \) for all \( 0 \leq i \leq j \) and some constant \( M > 0 \).

\begin{rem}
Similar to the argument above, one may also apply~\cite[Lemma~4.6]{jour:khalil2002nonlinear} to the system \( \dot{x} = -\Lc(z,t) + d(t) \) to establish that uniform exponential stability implies ISS, provided that \( \Lc(z,t) \) is continuous.
\end{rem}

\subsection{Time-delayed consensus}
As a final case, we consider consensus protocols with time delays, modeled by functional differential equations. These systems have been  
thoroughly examined in, e.g.,~\cite{Moreau2004stabconsensus, jour:Lu2017consensusdelays}, which establish sufficient conditions for asymptotic consensus in the presence of bounded communication delays. Other related works recently studying delayed second- and high-order consensus protocols are \cite{Gao2023secondorderdelay,li2024secondvaryinganddelays,TRINDADE2025delaymargins}.

Consider, as in~\cite[Lemma~3.1]{jour:Lu2017consensusdelays}, a system of the form
%In particular,~\cite[Lemma~3.1]{jour:Lu2017consensusdelays} considers systems of the form
\begin{equation}\label{eq:delayed_consensus}
    \dot{z}_i(t) = -\sum_{j \in \Nc_i} w_{i,j}(t) \left(z_i(t) - z_j(t - \tau_{i,j}(t)) \right),
\end{equation}
where \( w_{i,j}(t) \geq 0 \), and each delay \( \tau_{i,j}(t) \) is piecewise continuous and bounded: \( \tau_{i,j}(t) \leq \tau_{\max} < \infty \). Under the assumption that the time-integrated adjacency matrix 
\[
A(t) := \int_t^{t+T} W(\tau)\, \mathrm{d}\tau
\]
induces a \( \delta \)-digraph for some fixed $T>0$ that contains a fixed root node $k$ of a directed spanning tree for all $t\geq 0$, the system is known to reach consensus exponentially.

To express~\eqref{eq:delayed_consensus} compactly, we write it in functional form as
\begin{equation} \label{eq:delayed_consensus_vector}   
\dot{z}(t) = -D(t)z(t) + \mathcal{W}(z_t, t),
\end{equation}
where $D(t)$ is a diagonal matrix and \( z_t(\theta) := z(t + \theta) \) for \( \theta \in [-\tau_{\max}, 0] \), following standard notation in~\cite{jour:hale2013introductionfunctional}.

In this form, the operator \( \Lc_k(z_t,t) = -D(t)z(t) + \mathcal{W}(z_t,t) \) generally fails to satisfy Assumption~\ref{ass:kisinvariant}, as we show by counterexample in Appendix~\ref{app:counterexample}. Therefore, such a delayed operator cannot be used to define %as an intermediate component 
\( \Lc_k \) for \( k < n \). 
However, the delayed operator may be used for the outermost function in the composition~\( \Lc_n \), %as the final component \( \Lc_n \), 
since the unperturbed system~\eqref{eq:delayed_consensus}, for given delay functions, admits a unique solution. Our proof of Theorem~\ref{thrm:asymptoticconsensus} extends to this setting by interpreting~\( \xi_n \) as a continuous input signal to \( \xi_{n-1} \) in~\eqref{eq:statespace}. Provided that the remaining operators \( \Lc_k \), for \( k \neq n \), satisfy the assumptions of the theorem, the full solution \( z \) and its first \( n-1 \) derivatives are well defined and converge to consensus through the same inductive argument. A comprehensive treatment of functional differential equations can be found in~\cite{jour:hale2013introductionfunctional}.

\section{Case Studies}
\label{sec:applications}
We now explore applications of the compositional consensus framework developed in the previous sections. In particular, we focus on second-order formation control, where each agent is modeled as a double-integrator:
\begin{equation}
    \ddot{x} = u(x,t),
\end{equation}
with \( x(t) \in \mathbb{R}^N \). A general compositional consensus-based control law for this setting takes the form
\begin{equation}\label{eq:second-controller}
    u(x,t) = -\Lc_2(\dot{x} + \Lc_1(x,t), t) - \dtk{} \Lc_1(x,t).
\end{equation}
%This can be contrasted with the linear serial consensus controller,
This will be contrasted with a conventional second-order consensus protocol
\begin{equation}\label{eq:2ndconventional}
u_{\mathrm{conv}}(x,t) = -\Lc_{\mathrm{vel}} (\dot{x}(t)) - \Lc_{\mathrm{pos}} (x(t)).
\end{equation}
and what we will term a \emph{na\"ive} serial consensus:
\begin{equation}\label{eq:2ndserial}
    u_{\mathrm{ser}}(x,t) = -(\Lc_2 + \Lc_1)(\dot{x}(t)) - \Lc_2\circ \Lc_1 (x(t)),
\end{equation}

\begin{rem}
In many applications, the objective is to steer a group of agents into a fixed formation and maintain a constant collective velocity. This can be achieved by introducing a desired relative position vector \( d_{\mathrm{ref}} \in \RN \) and a reference velocity \( v_{\mathrm{ref}} \in \R \). To enforce the desired formation, one can work in the transformed coordinates \( \tilde{x} = x - d_{\mathrm{ref}} \). This transformation does not alter the system dynamics and thus preserves the control structure. The agents will then asymptotically coordinate in the frame \( \tilde{x} \), implying that \( |x_i(t) - x_j(t)| \to |d_i - d_j| \) as \( t \to \infty \). To ensure correct velocity tracking, a leader-follower structure may be employed, where the agents synchronize with a designated leader moving at velocity \( v_{\mathrm{ref}} \).
\end{rem}

\subsection{Time-varying graph Laplacians}
We now consider an application where both operators in the compositional consensus protocol are defined using time-varying Laplacians: \( \Lc_1(x,t) = L_1(t)x \) and \( \Lc_2(x,t) = L_2(t)x \). The resulting closed-loop system becomes
\begin{equation}\label{eq:2ndtimevarying}
    \ddot{x} = -L_2(t)\left(\dot{x} + L_1(t)x\right) - \dot{L}_1(t)x - L_1(t)\dot{x}.
\end{equation}
As established in Proposition~\ref{prop:timevarying}, if \( L_1(t) \) and \( L_2(t) \) are piecewise continuous and sufficiently connected over time, and if \( L_1(\cdot) \in C^0 \), then they can be used to construct a compositional consensus protocol that guarantees second-order consensus, via Theorem~\ref{thrm:asymptoticconsensus}.

Notably, such guarantees are generally not available for the corresponding na\"ive serial and conventional consensus protocols in~\eqref{eq:2ndserial} and~\eqref{eq:2ndconventional}, when the Laplacians are time-varying. We consider the following example of a string formation with time-varying connectivity.

\begin{example}
We define \( L_{\mathrm{path}} \in \mathbb{R}^{N \times N} \) as the Laplacian of a directed path graph with the following structure:
\[
L_{\mathrm{path}} = 
\bmat{
0 & 0 & &  \\
-1 & 1  \\
& \ddots & \ddots & \\
& & -1 & 1  \\
}.
\]
Let \( L_1(t) = L_2(t) = D(t)L_{\mathrm{path}} \), where \( D(t) \) is a time-varying diagonal matrix defined as
\[
[D(t)]_{i,i} = \max\left\{ \sin(\omega_i t + \phi_i), 0 \right\},
\]
with individual frequency \( \omega_i \neq 0 \) and phase \( \phi_i \in [0, 2\pi) \). This choice satisfies all conditions for Theorem~\ref{thrm:asymptoticconsensus} and ensures that both \( L_1(t) \) and \( L_2(t) \) are connected over time.

Figure~\ref{fig:compare_serial_conv} shows a simulation of a second-order vehicle formation with \( N = 20 \) agents under this protocol, with randomly chosen frequencies and phases. Despite the complexity of the system and the time-varying graph structure, the compositional protocol successfully coordinates the agents. It achieves second-order consensus, with the agents converging to their desired relative positions.

For comparison, we simulate the same formation using the time-varying versions of the conventional and na\"ive serial consensus controllers. The conventional controller is defined~as
\[
u_{\mathrm{conv}}(x,t) = -L_1(t)\dot{x} - L_2(t)x,
\]
and the result is shown in Figure~\ref{fig:timevaryingconv}. The na\"ive serial consensus protocol is given by
\[
u_{\mathrm{ser}}(x,t) = -(L_2(t) + L_1(t))\dot{x} - L_2(t)L_1(t)x,
\]
with the corresponding result shown in Figure~\ref{fig:timevaryingserial}.
Both alternative controllers exhibit poor transient performance: the na\"ive serial consensus has a slow and oscillatory convergence to the reference trajectory, 
while the conventional controller produces extreme oscillations.
\end{example}

\begin{figure*}[!t]
    \centering
    \begin{subfigure}{.32\linewidth}
        \centering
        \includegraphics[width=\linewidth]{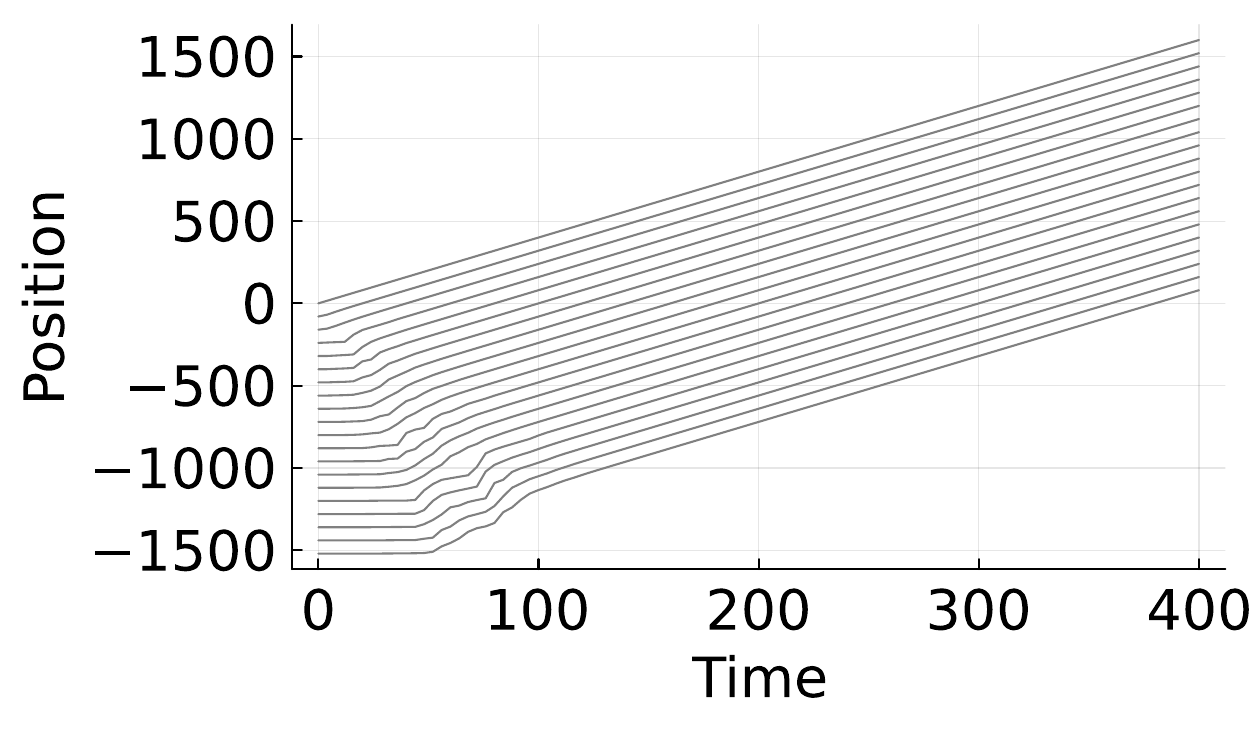}
        \caption{Compositional consensus}
        \label{fig:timevaryingcompos}
    \end{subfigure}
    \hfill
        \begin{subfigure}{.32\linewidth}
        \centering
        \includegraphics[width=\linewidth]{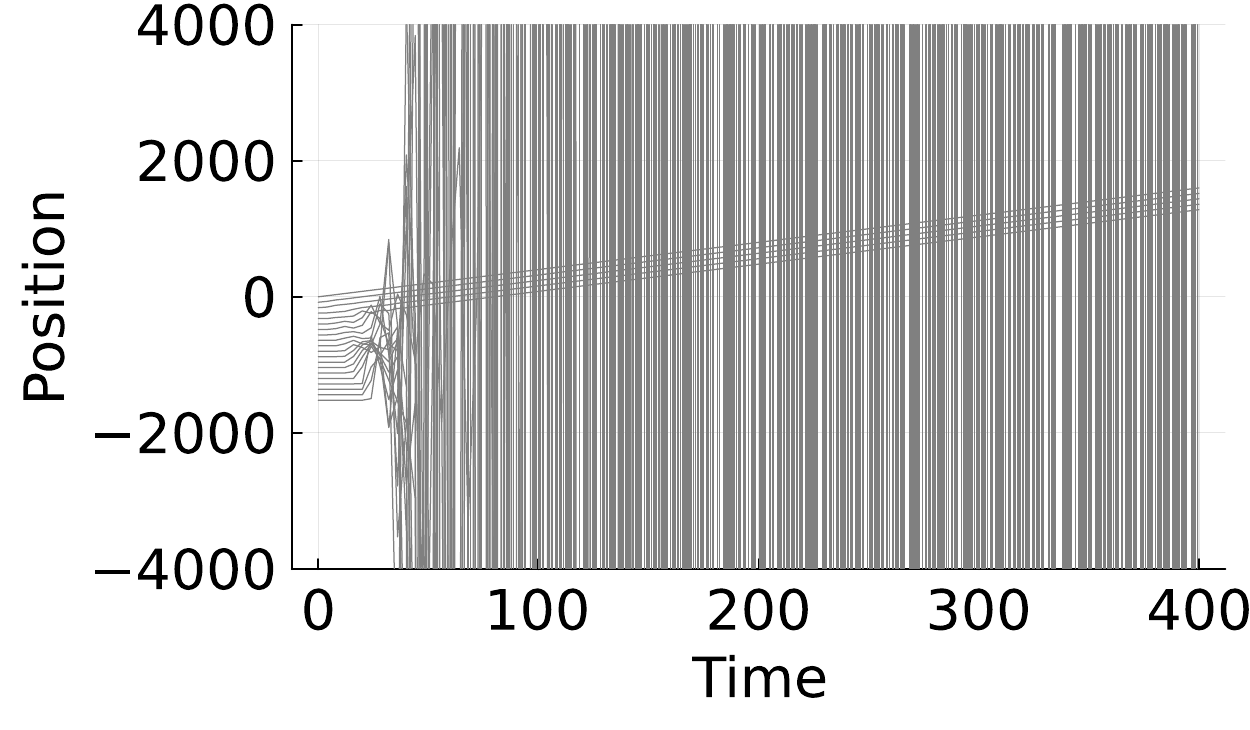}
        \caption{Conventional consensus}
        \label{fig:timevaryingconv}
    \end{subfigure}
       \hfill
    \begin{subfigure}{.32\linewidth}
        \centering
        \includegraphics[width=\linewidth]{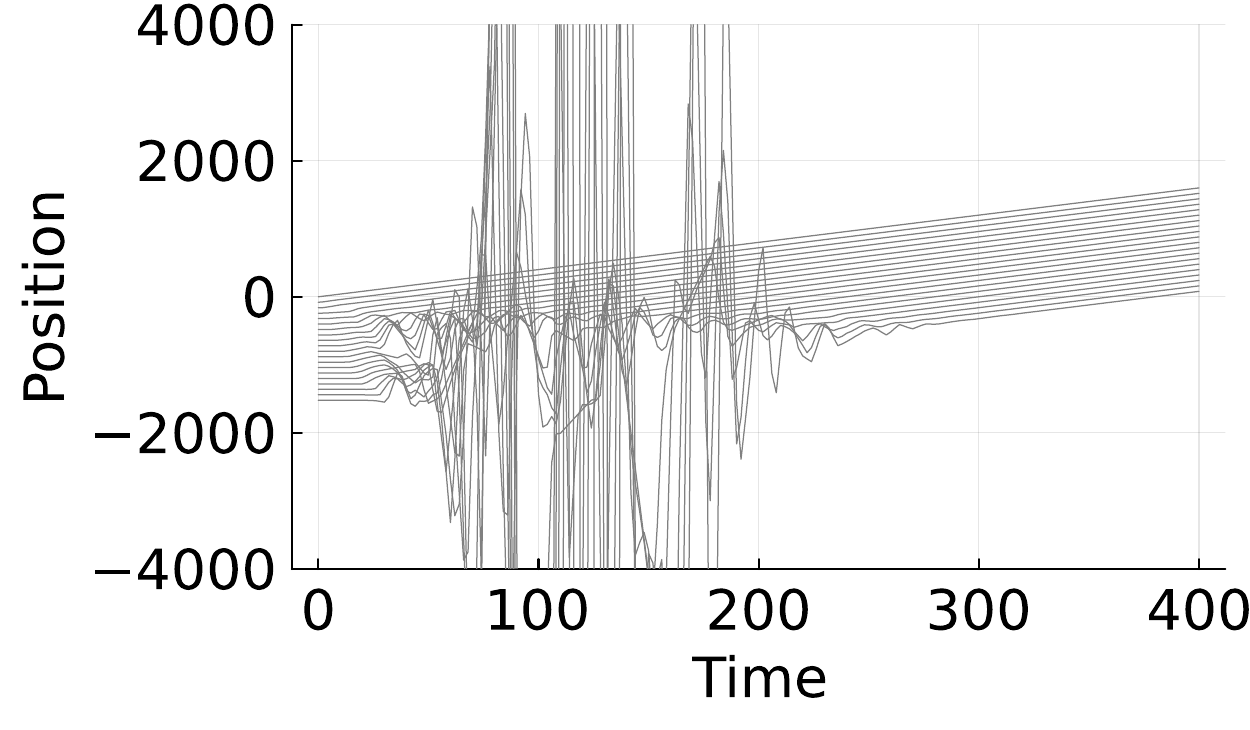}
        \caption{Na\"ive serial consensus}
        \label{fig:timevaryingserial}
    \end{subfigure}
    \caption{Simulation results using time-varying graph Laplacians whose connection strengths vary sinusoidally. The compositional controller achieves stable second-order consensus. The {na\"ive} serial consensus exhibits a significant transient but eventually converges, while the conventional consensus seems to be truly unstable.}
    \label{fig:compare_serial_conv}
\end{figure*}

\subsection{Saturated coordination}
\label{sec:satcoord}
Time-varying consensus protocols represent just one class of systems that benefit from compositional consensus. Another important and challenging class involves nonlinear protocols, particularly those incorporating input saturation. As established in Proposition~\ref{prop:saturation}, the compositional consensus framework can accommodate operators such as
\[
\Lc_1(x,t) = \mathrm{sat}(L_1 x), \quad \Lc_2(x,t) = \mathrm{sat}(L_2 x),
\]
provided that both Laplacians contain a directed spanning tree. By contrast, no general guarantees exist for the conventional or even the na\"ive serial consensus protocols when such nonlinearities are present. Consider the following example.

\begin{example}
Consider formation control over a directed string network. That is, the case where \( \Lc_1(x,t) = \Lc_2(x,t) = \mathrm{sat}(L_{\mathrm{path}} x) \). The resulting control law becomes
\[
u(x,t) = -\mathrm{sat}\left( L_{\mathrm{path}}( \dot{x} +\mathrm{sat}(L_{\mathrm{path}} x)) \right) - \dtk{}  \mathrm{sat}(L_{\mathrm{path}} x).
\]
The formation is initialized with a nonzero positional error to highlight the effect of saturation. The simulation results are shown in Figure~\ref{fig:saturatedcoord}. Despite the nonlinearities, the compositional controller (Figure~\ref{fig:satcompos}) achieves a smooth transition to second-order consensus.

For comparison, we simulate the same system under saturated versions of the na\"ive serial and conventional consensus controllers. That is
\[
    u_\mathrm{ser}(x,t)=-2\mathrm{sat}(L_\mathrm{path}\dot{x})-\mathrm{sat}(L_\mathrm{path}\mathrm{sat(}L_\mathrm{path} x)),
\]
and
\[
    u_\mathrm{conv}(x,t)=-\mathrm{sat}(L_\mathrm{path}\dot{x}) -\mathrm{sat}(L_\mathrm{path}x).
\]
The results are shown in Figures~\ref{fig:satserial} and~\ref{fig:satconv}, respectively. The na\"ive serial consensus shows similar but slightly slower convergence than the compositional consensus. The saturated conventional consensus shows an indication of string instability. We have conducted larger simulations that verify this indication. 
\end{example}

\begin{figure*}[]
    \centering
    \begin{subfigure}{.31\linewidth}
        \centering
        \includegraphics[width=\linewidth]{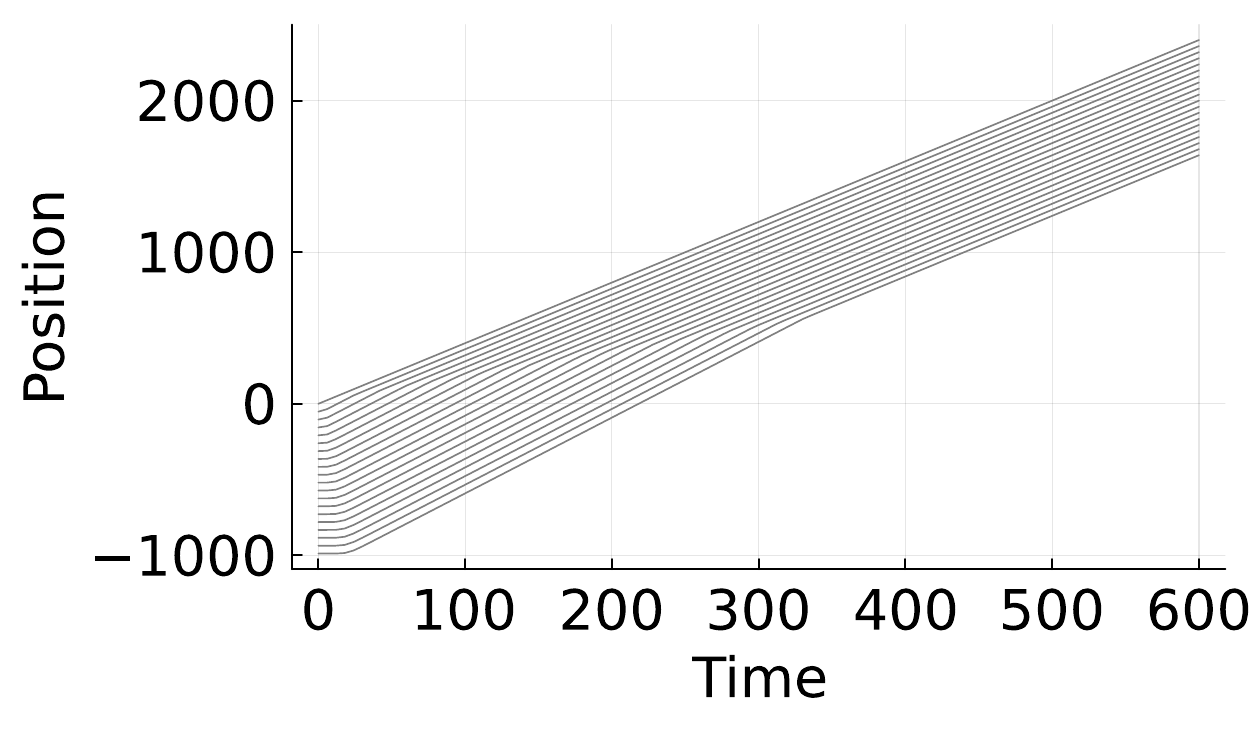}
        \caption{Compositional consensus}
        \label{fig:satcompos}
    \end{subfigure}
    \hfill
    \begin{subfigure}{.31\linewidth}
        \centering
        \includegraphics[width=\linewidth]{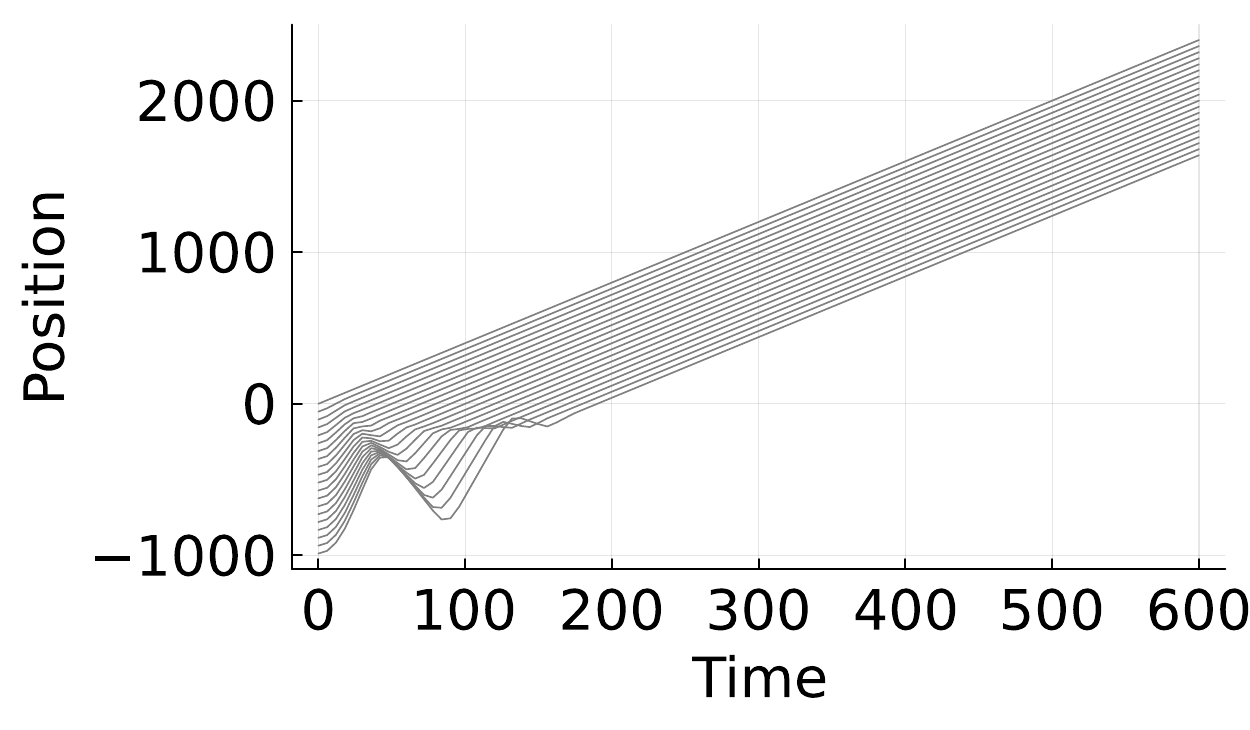}
        \caption{Conventional consensus}
        \label{fig:satconv}
    \end{subfigure}
    \hfill
    \begin{subfigure}{.31\linewidth}
        \centering
        \includegraphics[width=\linewidth]{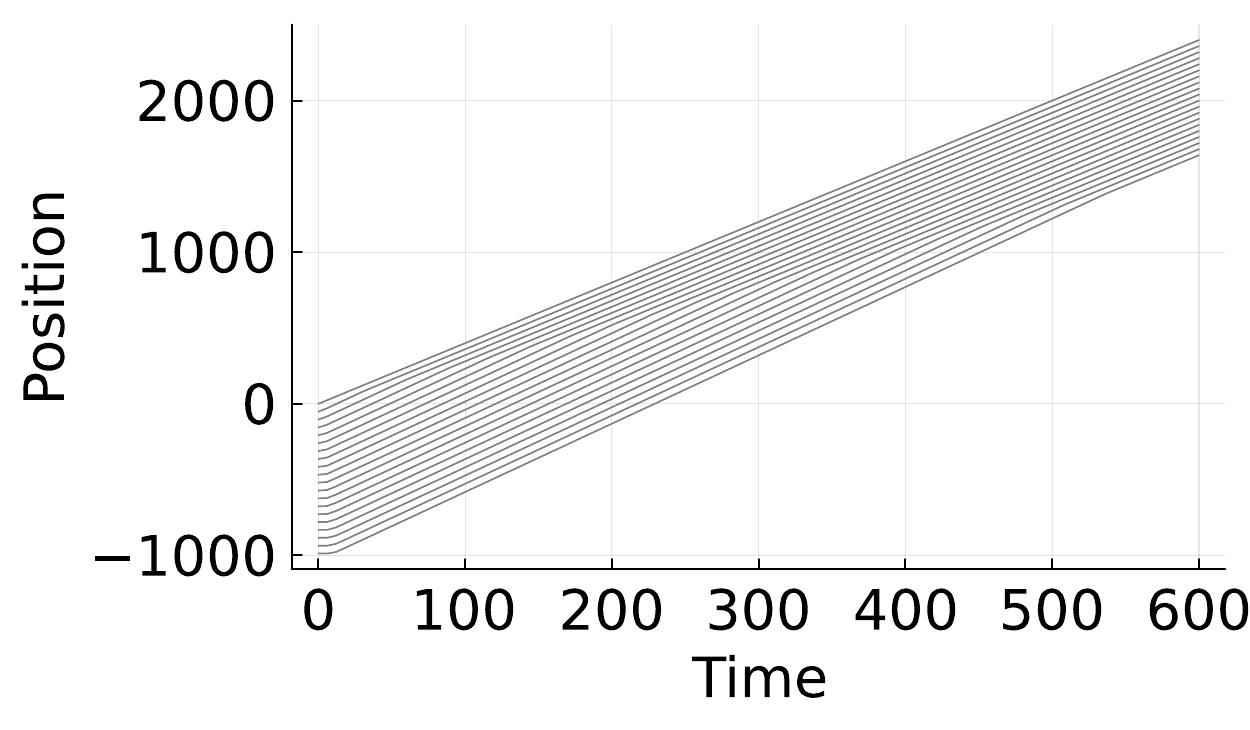}
        \caption{Na\"ive serial consensus}
        \label{fig:satserial}
    \end{subfigure}
    \caption{Simulation of compositional, conventional, and na\"ive serial consensus under saturated control inputs. The compositional and na\"ive serial consensus protocols achieve smooth convergence, while conventional consensus exhibits an undesired transient indicative of string~instability. }%indicate potential string instability.}
    \label{fig:saturatedcoord}
\end{figure*}

\subsection{Delayed absolute feedback}
We conclude with an application involving delayed absolute feedback, such as GPS-based velocity measurements. In other words, we now consider a case where Assumption~\ref{ass:kisinvariant} is relaxed. %Other related works that have recently studied delayed second- and high-order consensus protocols are \cite{Gao2023secondorderdelay,li2024secondvaryinganddelays,TRINDADE2025delaymargins}. 
In vehicle platoons, absolute feedback has been proposed to improve performance; however, in practice, it would typically be %it is often 
received aperiodically and with uncertain delays. This example investigates such a scenario.

\begin{example}
We consider a delayed consensus protocol based on absolute measurements with static coupling weights. Written in individual-agent form, the dynamics are
\[
\dot{x}_i(t) = -d_i\left(x_i(t) - x_{\mathrm{GPS}}(t - \tau_i(t))\right),
\]
which can be compactly expressed as
\[
\Lc_2(x_t, t) = D x(t) - D_{\tau(t)}(\mathbf{1} x_{t,\mathrm{GPS}}),
\]
where \( x_t(\theta) := x(t + \theta) \) for \( \theta \in [-\tau_{\max}, 0] \), and \( D \) is a diagonal matrix of feedback weights.

For the other operator, we use a standard linear time-invariant consensus protocol: \( \Lc_1(x,t) = L_{\mathrm{path}} x(t) \). The first row of \( L_{\mathrm{path}} \) is all zeros, modeling a virtual leader. The resulting compositional control law is
\[
u_{\mathrm{comp}}(x_t, t) = -D(\dot{x} + L_{\mathrm{path}} x) + D_{\tau(t)}(\mathbf{1} \dot{x}_{t,\mathrm{GPS}}) - L_{\mathrm{path}} \dot{x}.
\]
Defining \( e = L_{\mathrm{path}} x \), and rearranging terms, the individual-agent control becomes
\[
u_{i,\mathrm{comp}}(x_t, t) = -d_i \left( \dot{x}_i(t) - \dot{x}_{\mathrm{GPS}}(t - \tau_i(t)) \right) - d_i e_i(t) - \dot{e}_i(t).
\]
The last two terms correspond to standard relative feedback with local neighbors, while the first involves delayed absolute velocity feedback. This term can be rewritten as
\[
\left( \dot{x}_i(t) - \dot{x}_i(t - \tau_i(t)) \right) - \left( \dot{x}_{\mathrm{GPS}}(t - \tau_i(t)) - \dot{x}_i(t - \tau_i(t)) \right),
\]
which separates into two interpretable components:  
1) the change in the agent’s velocity since the last measurement, and  
2) a delayed relative velocity signal received from the GPS.

In practice, each agent stores a record of its past velocity and periodically receives delayed GPS-based velocity references. This allows the required feedback to be implemented despite asynchronous and uncertain communication delays.

We simulate a vehicle formation with \( d_i = 1 \) for all agents and delays \( \tau_i(t) \) sampled from a Poisson process with a mean inter-arrival time of 1 second. The compositional controller is compared to:  
1) a conventional consensus protocol with perfect (non-delayed) absolute velocity feedback, and 2) the same conventional consensus protocol subject to the same delays $\tau_i(t)$ as in the compositional case.

Figure~\ref{fig:delayedcoordination} shows that the compositional consensus protocol achieves smooth second-order coordination despite the delays. The conventional controller also performs well with ideal feedback, but its performance degrades significantly under delay, resulting in oscillatory behavior. 
% \emma{But the comparison is a little unfair: what if the conventional also makes use of delayed absolute \emph{position} feedback? No need to change, but let's be prepared to defend comparison? }
\end{example}

\begin{figure*}[]
    \centering
    \begin{subfigure}{.31\linewidth}
        \centering
        \includegraphics[width=\linewidth]{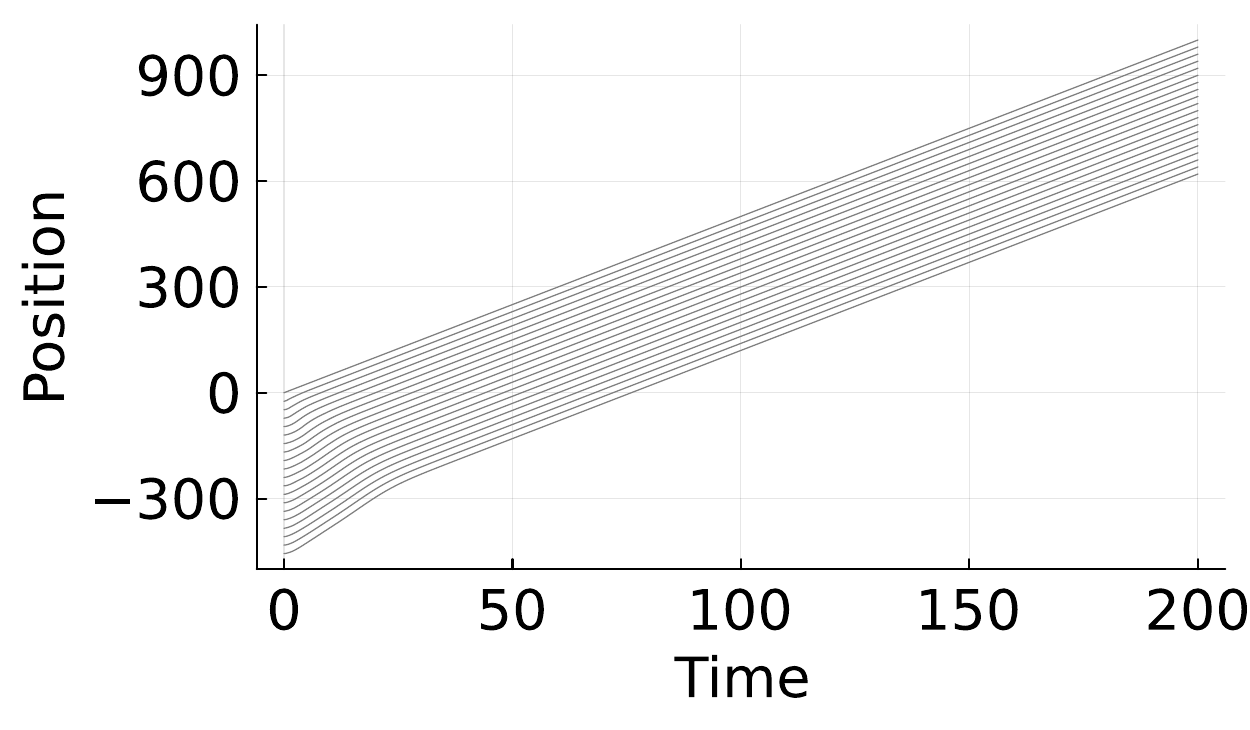}
        \caption{Compositional consensus with delayed velocity feedback}
        \label{fig:satcompos}
    \end{subfigure}
    \hfill
    \begin{subfigure}{.31\linewidth}
        \centering
        \includegraphics[width=\linewidth]{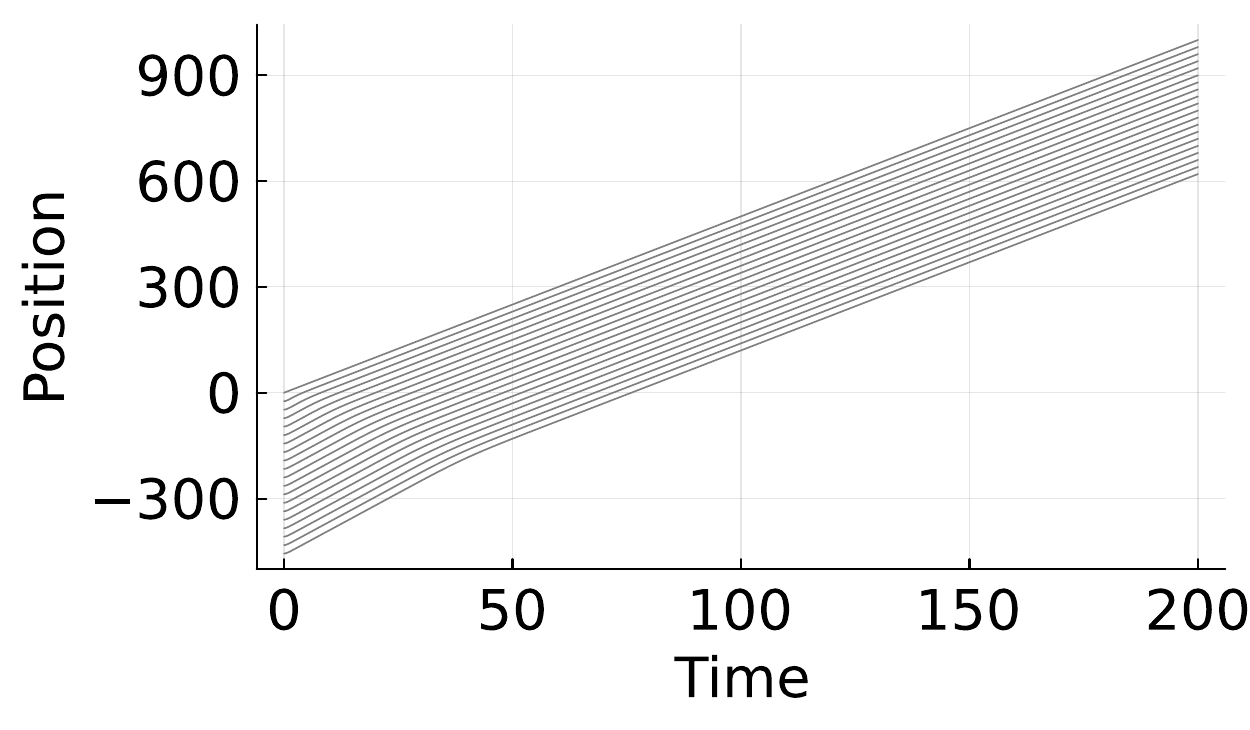}
        \caption{Conventional consensus with ideal velocity feedback}
        \label{fig:satserial}
    \end{subfigure}
    \hfill
    \begin{subfigure}{.31\linewidth}
        \centering
        \includegraphics[width=\linewidth]{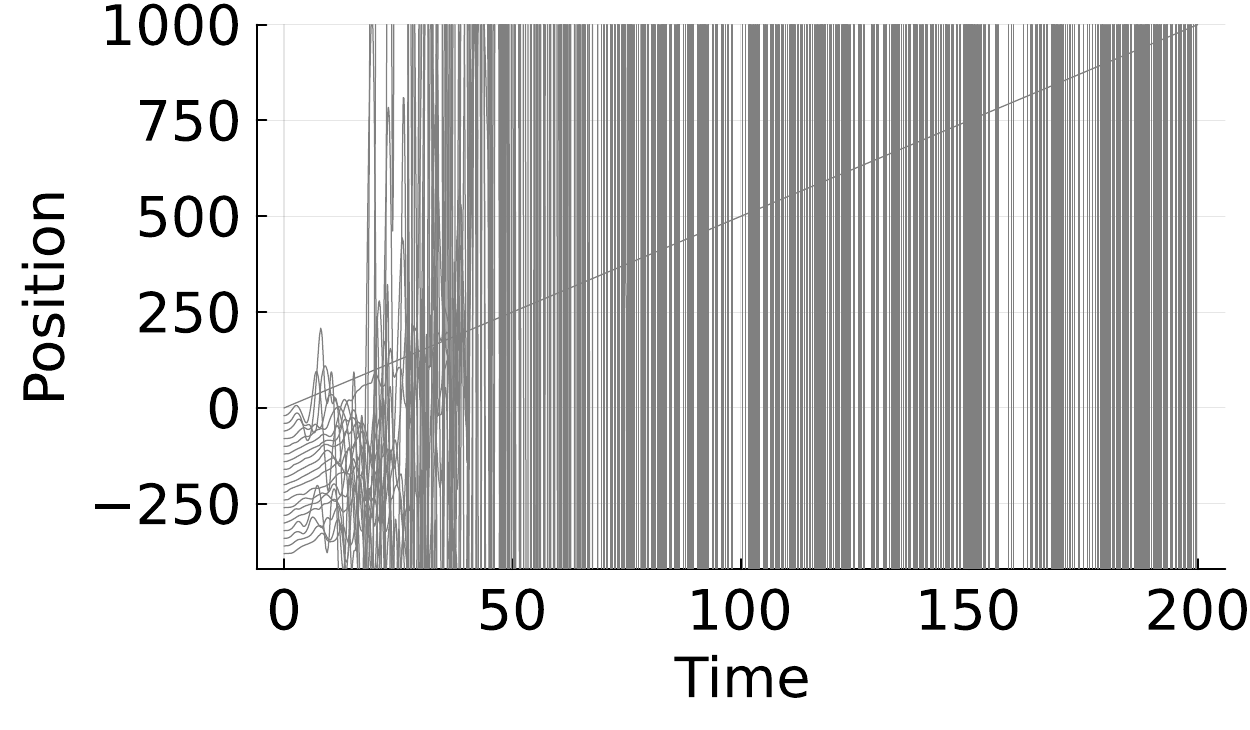}
        \caption{Conventional consensus with delayed velocity feedback}
        \label{fig:satconv}
    \end{subfigure}
    \caption{Simulation of second-order consensus under absolute velocity feedback with and without delay. The compositional consensus controller is robust to feedback delays. This is unlike the conventional controller, which shows oscillatory behavior under the same conditions.}
    \label{fig:delayedcoordination}
\end{figure*}

\section{Conclusions}\label{sec:conclusions}

In this work, we expanded the theory of high-order coordination by introducing and analyzing a general framework for \emph{compositional consensus}. This approach provides a flexible and modular design paradigm that accommodates practical challenges such as time-varying dynamics, nonlinearities, and communication delays. In particular, it allows us to build upon the rich literature on the convergence of first-order consensus under various non-ideal conditions and immediately apply them for higher-order formation control. We focused on second-order coordination and vehicular formations as motivating examples, but the framework is broadly applicable. Potential applications include frequency coordination in power systems, temperature regulation in district heating networks, and large-scale multi-agent systems such as drone swarms. Implementing the controller in general presumes signaling in an $n$-hop neighborhood, where $n$ is the order of the local integrator dynamics. We remark, though, that in the case $n = 2$, a nearest-neighbor implementation is also possible through a ``look-ahead and look-behind protocol,'' see~\cite{jour:hansson2023scalable}. 

The main theoretical contribution of our work is a set of sufficient conditions for achieving asymptotic coordination in the compositional setting, formally stated in Theorem~\ref{thrm:asymptoticconsensus}. These results extend the reach of classical consensus theory and offer tools for principled design in complex settings.

Future work may involve identifying necessary conditions for coordination and extending current string-stability and scalability results, which are already established for linear serial consensus, to nonlinear and time-varying compositional designs. In particular, to investigate performance guarantees that are uniform in network size.  Such extensions are essential for robust and scalable deployment in large coordinated systems.

\appendix
% \section{Proof of Technical Lemmas}\label{app:lemmaproofs}
\subsection{Proof of Lemma~\ref{lem:consensusequivalnce}}\label{app:conensusequivalenceproof}
To aid the presentation of the proof, we begin with the following supporting lemma.

\begin{lemma}\label{lem:consensus_invariance}
    Let $\Lc\in C^n$ and satisfy $\Lc(z+\1 a(t),t)$ for any integrable function $a(t)$, then, for any $k\leq n$ the following is true: $\dtk{k}\Lc(z,t)=B_k(z,\dot{z},\dots,z^{(k)},t)=B_k(z-\1 b_0(t),\dot{z}-\1 b_1(t),\dots,z^{(k)}-\1 b_k(t),t)$ for any integrable functions $b_1, b_2,\dots, b_k$.
\end{lemma}
\begin{proof}
    We prove this by induction. The base case is directly proven by $B_0(x)=\Lc(x,t)$. Now, suppose it is valid for $k\leq n-1$. Firstly, taking the partial derivatives for any $ j \leq k$: $\partial_{z^{(j)}} B_k(z,\dots,z^{(k)},t)=\partial_{z^{(j)}} B_k(z-\1 b_0(t),\dots,x^{(k)}-\1 b_k(t),t)$. This shows that all partial derivatives are invariant to arbitrary translation along the consensus, and the same argument also holds for the partial time derivative. Now, taking the time derivative of the consensus translated  equation results in:
    \begin{multline*}
        \dtk{k+1}\Lc(x,t)=\partial_t B_k(x-\1 c_0(t),\dots,x^{(k)}-\1 c_k(t),t)+\\ \sum_{j=0}^k \!\partial_{x^{(j)}} B_k(x-\1 c_0(t),\dots,x^{(k)}-\1 c_k(t),t) (x^{(j+1)}-\1\dot{c}_j(t))
    \end{multline*}
Let $c_j(t)=\int_0^tb_{j+1}(\tau)\mathrm{d}\tau$ and for each partial derivative term do the outlined translation to get
\begin{multline*}
        \dtk{k+1}\Lc(x,t)=\partial_t B_k(x-\1 b_0(t),\dots,x^{(k)}-\1 b_k(t),t)+\\ \sum_{j=0}^k\! \partial_{x^{(j)}}\! B_k(x-\1 b_0(t), \!... ,x^{(k)}\!\!-\1 b_k(t),t) (x^{(j+1)}\!\!-\1 b_{j+1}(t)).
    \end{multline*}
    This shows the sought translation invariance and concludes the proof.
\end{proof}

We now proceed to the proof of Lemma~\ref{lem:consensusequivalnce}.

\begin{proof}
    First, we establish that the initial condition of $\xi_k$ is uniquely determined by the initial condition $x$  and its first $n-1$ derivatives. By using the relation $x=\xi_1$ and \eqref{eq:statespace}, the following relation can be derived
    $$x^{(j)}=-\dtk{j-1}\Lc_1(\xi_1,t)-\dots-\Lc_j(\xi_j,t)+\xi_{j+1}.$$
    Due to the smoothness of $\Lc_k$, it is possible to expand the time derivatives in terms of the partial derivatives through the chain rule. The time derivatives can thus be expressed as
    $$\dtk{j}\Lc_k(\xi_k,t)=B_{k,j}(\xi_k,\dot{\xi}_k,\dots,\xi_k^{(j)},t).$$
    Since $\dot{\xi}_k=-\Lc_k(\xi_k,t)+\xi_{k+1}$, it is possible to reduce the derivative dependence recursively and to show that 
    % \emma{I felt I wanted a middle step here?}
    $$\dtk{j}\Lc_k(\xi_k,t)=\hat{B}_{k,j}(\xi_k,\xi_{k+1},\dots,\xi_{k+j},t).$$
    Applying this to the general case leads to 
    $$
        \xi_{j+1}=x^{(j)}+\sum_{k=1}^j \hat{B}_{k,j-k}(\xi_k,\dots,\xi_j,t), 
    $$
    for $j = 0,\ldots,n-1$. Evaluating this at $t=0$ shows that $\xi_{j+1}$ is uniquely determined by $x^{(j)}(0)$ and the initial conditions of $\xi_k(0)$ for $k\leq j$. This, together with $\xi_1(0)=x(0)$, can be used to prove that $\xi_k(0)$ is uniquely determined by the initial condition of $x$ and its derivatives through a simple induction hypothesis. An analogous proof can be made in the reverse direction and, therefore, is omitted.

    For the second part of the proof, we will show that consensus of the states $\xi_k$ implies that $x$ achieves $n\ts{th}$-order consensus. As induction hypothesis, assume that $\|\xi_k^{(j)}-\1 a_{k+j}(t)\|\to 0$, with the induction step taken in the $j$ direction. The base case follows from the assumption that $\xi_k$ all reach a consensus, that is, $\|\xi_k-\1 a_k(t) \|\to 0$.

    For the induction step consider the general expression for $\xi_k^{(j+1)}$ for $k+j\leq n-1$, which is
$$\xi_k^{(j+1)}=-\dtk{j}\Lc_k(\xi_k,t)+\dtk{j}\xi_{k+1}.$$
Using Lemma~\ref{lem:consensus_invariance}, this can be expressed in terms of the translated states
    $$\xi_k^{(j+1)}=-B_{k,j}(\xi_k-\1 a_k(t),\dots,\xi_k^{(j)}-\1 a_{k+j},t)+\xi_{k+1}^{(j)}.$$
    By the premise of the theorem, $\|B_{k,j}\|$ can be bounded by $\alpha_{k,j}\in \mathcal{K}$. Subtracting $\1 a_{k+j+1}$ on both sides, taking the norm, using the triangle inequality, and bounding using the class $\mathcal{K}$ function $\alpha_{k,j}$ leads to
\begin{multline*}
    \|\xi_k^{(j+1)}-\1 a_{k+j+1}\|\leq \|\xi_{k+1}^{(j)}-\1 a_{k+j+1}(t)\| +\\\alpha_{k,j}\left(\max\left\{\|\xi_k-\1 a_k(t)\|,\dots,\|\xi_k^{(j)}-\1 a_{k+j}\|\right\}\right)
    \end{multline*} 
    Now, taking the limits on both sides, using the induction hypothesis together with the continuity of $\alpha_{k,j}$ shows that $\lim_{t\to \infty}\|\xi_k^{(j+1)}-\1 a_{k+j+1}(t)\|=0$. Thus $\xi_k$ achieves an $(n-k+1)$\ts{th}-order consensus. Since $x(t)=\xi_1(t)$ it follows that $x$ achieves an $n$\ts{th}-order consensus.

    The other direction, that is, $x$ achieving $n\ts{th}$ order consensus implying that $\xi_k$ achieves consensus is conducted similarly. Now, the induction hypothesis is that $\xi_{k}^{(j)}-\1 a_{j+k}(t)$ where this will be proved by induction steps in $k$.
    First, using the relation of $x(t)=\xi_1(t)$ shows that $\xi_1$ achieves $n$\ts{th} order consensus. The $n$\ts{th} order consensus implies that $\|\xi^{(j)}-\1 a_j(t)\| \to 0$ for some functions $a_j(t)$. For the induction step, we consider the relation
$$\dot{\xi}_k=-\Lc_k(\xi_k,t) +\xi_{k+1}$$
This can be differentiated $j\leq n-k-1$ times, and be rearranged to
$$\xi_{k+1}^{(j)}-\1 a_{j+k+1}=\xi^{(j+1)}_k-\1 a_{j+k+1}+\dtk{j}\Lc_k(\xi_k,t).$$
Now, Lemma~\ref{lem:consensus_invariance} is used to express $\dtk{j} \Lc_k$ in terms of $B_{k,j}$ and in particular in the translated states 
$$ \dtk{j} \Lc_k(\xi_k,t)=B_{k,j}(\xi_k-\1 a_k(t),\dots,\xi_k^{(j)}(t)-\1 a_{k+j}(t),t). $$
Applying the triangle inequality, bounding $\|B_{k,j}\|$ with $\alpha_{k,j}$, and concluding by taking the limit shows that
$$\lim_{t\to \infty}\|\xi_{k+1}^{(j)}-\1 a_{j+k+1}\|=0.$$
This proves the induction step. Since this also shows that $\|\xi_{k}-\1 a_k(t)\| \to 0$, we can conclude that the states will achieve consensus.
\end{proof}

\subsection{Proof of Lemma~\ref{lem:statespaceconsensus}}\label{app:lemstatespaceconsensus}
\begin{proof}
    The existence and uniqueness of $\xi_n$ is part of the lemma premise. For the remaining states it is simple to verify that Assumption~\ref{ass:kislipschitz} implies that \eqref{eq:statespace} is globally Lipschitz in $\xi_k$ and piecewise continuous in $t$. Existence and uniqueness follow from a standard application of Carathéodory's existence and uniqueness theorem.

    Through induction, we'll prove that consensus will be reached, that is $\|\xi_k-\1a_k(t)\|\to 0$ for some functions $a_k(\cdot)$. The base case with $\|\xi_n-\1 a_n(t)\|\to 0$ follows from our assumption. Suppose it is true for all $\xi_k$ where $k\geq p+1$. The solution for $\xi_p$ satisfy
$$\dot{\xi}_p(t)=-\Lc_p(\xi_p(t),t)+\xi_{p+1}(t)$$
Now, subtracting the asymptotic solution of $\xi_{p+1}$ from both sides and using the fundamental theorem of calculus results in
\begin{multline*}
    \dtk{}\left(\xi_p(t)-\1\int_0^t a_{p+1}(\tau)\mathrm{d}\tau\right)\\=-\Lc_p\left(\xi_p(t)-\1\int_0^t a_{p+1}(\tau)\mathrm{d}\tau,t\right)+\xi_{p+1}-\1 a_{p+1}(t),
    \end{multline*}
where the invariance of $\Lc_p$ through Assumption~\ref{ass:kisinvariant} was also used. Let $z_p(t)=\xi_p(t)-\1\int_0^t a_{p+1}(\tau)\mathrm{d}\tau$ and $w_p(t)=\xi_{p+1}-\1 a_{p+1}(t)$. Then  $z_p$ satisfies
$$\dot{z}_p=-\Lc_p(z_p,t)+w_p(t)$$
where $\|w_p(t)\|\to 0$, which allows us to apply Assumption~\ref{ass:kislocallyISS}. In particular, there is a time $T_p$ such that $\|w_p(t)\|\leq M_p$, where this system is ISS for some seminorm $\vertiii{\cdot}$. To assert that $\lim_{t\to \infty} \|z_p(t)-\1 b_p(t)\|=0$ we can use the $\epsilon$ and $T$ definition for the limit. For any $\epsilon>0$, choose $T_p'$ such that $\|w_p(t)\|<\gamma_k^{-1}(\epsilon/2)$ for all $t> T_p'$. Now, using the ISS property starting at $T_p'$, we get
$$\vertiii{z_p(t)}< \beta_p(\vertiii{z_p(T_p')},t)+\frac{\epsilon}{2}$$
By the definition of $\beta_p$ it's possible to choose a time $T\geq T_p'$ such that $\beta_p(\vertiii{z_p(T_p')},t)<\epsilon/2$.
This proves that the seminorm converges and in particular that there exists a $b_p(t)$ such that $\|z_p(t)-\1 b_p(t)\|\to 0\implies \|\xi_p(t)-\1\left(b_p(t)+\int_0^t a_{p+1}(\tau)\mathrm{d}\tau\right)\|\to 0$, letting $a_p=b_p(t)+\int_0^t a_{p+1}(\tau)\mathrm{d}\tau$ concludes the induction step.
\end{proof}

\subsection{Proof of Proposition~\ref{prop:saturation}}\label{app:proof_saturation}
\begin{proof}
    We begin by proving the result for a leader-follower network. In this case, the dynamics can be rewritten as
    \[
        L\dot{z} = -L\,\mathrm{sat}(z) + L d(t).
    \]
    This representation follows from left-multiplying the system by the Laplacian \( L \). While \( L \) is singular, we are only interested in the evolution of the disagreement vector \( e = Lz \), which remains orthogonal to the consensus subspace.

    A leader-follower network has a unique leader agent whose state remains unaffected by the others. For this agent, \( [Lz]_i = 0 \) for all time. Our goal is to show that, under sufficiently small disturbances, all other agents enter and remain in the linear regime, i.e., \( |[Lz]_i(t)| < 1 \) for all \( t \geq T_i \).

    We proceed by induction along a simple directed path of influence from the leader to any follower. Let the path consist of \( m+1 \) agents labeled \( p_0, p_1, \ldots, p_m \), with \( p_0 \) being the leader.

Base case: The leader agent satisfies \( e_{p_0}(t) = 0 \) for all \( t \), so it is trivially in the linear regime.

Inductive step: Suppose that agent \( p_k \) satisfies \( |e_{p_k}(t)| < 1 \) for all \( t \geq T_k \). We aim to show that agent \( p_{k+1} \) enters the linear regime in finite time \( T_{k+1} \).

The dynamics of agent \( p_{k+1} \) are given by
\[
\dot{e}_{p_{k+1}} = d_{p_{k+1}} - \sum_{j \in \mathcal{N}_{p_{k+1}}} w_{p_{k+1},j} \left(\mathrm{sat}(e_{p_{k+1}}) - \mathrm{sat}(e_j) \right).
\]
Since agent \( p_k \) is in the linear regime after time \( T_k \), we have \( \mathrm{sat}(e_{p_k}) = e_{p_k} \) for \( t \geq T_k \). Without loss of generality, assume \( e_{p_{k+1}}(T_k) > 0 \) (the argument is symmetric for the negative case). We upper-bound the dynamics as
\[
\dot{e}_{p_{k+1}} \leq |d_{p_{k+1}}| + w_{p_{k+1},p_k} |e_{p_k}| - w\,\mathrm{sat}(e_{p_{k+1}}) + w - w_{p_{k+1},p_k},
\]
where \( w = \sum_j w_{p_{k+1},j} \) is the total weight of incoming edges to agent \( p_{k+1} \).

Now, if the disturbance is sufficiently small so that
\[
|d_{p_{k+1}}| < w_{p_{k+1},p_k}(1 - |e_{p_k}|),
\]
then the right-hand side of the inequality becomes negative whenever \( e_{p_{k+1}} \geq 1 \), implying that the agent must enter the region \( |e_{p_{k+1}}| < 1 \) in finite time \( T_{k+1}' \).

After entering the linear regime, the dynamics simplify, and the state can be upper-bounded by a linear system with equilibrium state
\[
e_{p_{k+1}}^{*} = \frac{|d_{p_{k+1}}| + w_{p_{k+1},p_k} |e_{p_k}| + w - w_{p_{k+1},p_k}}{w},
\]
which can be made strictly less than 1 by choosing \( d_{p_{k+1}} \) sufficiently small. The state will converge towards this bound and reach any point between this and $1$ in finite time \( T_{k+1} > T_{k+1}' \), and then remain there for all future time. If the agent started below this steady-state bound, the same conclusion holds with \( T_{k+1} = T_k \).
This completes the inductive step.

Since all agents are connected by a finite directed path originating from the leader, each agent enters the linear regime in finite time. Once all agents lie within the linear region, the local ISS result from Proposition~\ref{prop:serial} can be applied to show that the disturbance \( d(t) \) has a bounded effect on \( \|Lz(t)\|_\infty \).

    To prove the general case, it suffices to show that the agents within the unique strongly connected component (SCC) of the graph underlying $L$ enter and remain in the linear regime for all $t\geq T$ provided the input $d(t)$ is sufficiently small. By definition, the agents in this component evolve independently of the remaining agents. 

    Without loss of generality, consider the subgraph corresponding to the SCC, with Laplacian $\tilde{L}\in \R^{K\times K}$, where $K\geq 2$. Since the subgraph is strongly connected, its zero eigenvalue is simple, and the corresponding left Perron eigenvector $w$ can be chosen to have strictly positive entries.

    Define the diagonal matrix $W=\mathrm{diag}(w)$. Then, the matrix $\tilde{L}^\top W$ satisfies $\tilde{L}^\top W\1=0,$ i.e., it is a graph Laplacian of a strongly connected graph.

    Consider the Lyapunov function
    $$V(t)=\dfrac{\tilde z^\top \tilde{L}^\top W \tilde z}{2}$$
    This function is non-negative and satisfies $V(t)=0 \iff \tilde z\in \mathrm{span}(\1)$, i.e., consensus. 

    Define $\tilde e=\tilde L \tilde z.$ Then, the time derivative of $V$ along trajectories of the system is 
    $$\dot{V}=-\tilde{z}^\top \tilde{L}^\top W \mathrm{sat}(\tilde{L}z)+\tilde{z}^\top \tilde{L}^\top \tilde{d}(t).$$
    Expanding this, and defining $\tilde e=\tilde L \tilde{z} $, we get
    $$\dot{V}(t)=\frac{\dot{V}(t)}{2}=-\sum_{i=1}^K w_i|\tilde{e}_i|(|\mathrm{sat}(e_i)|-\tilde{d}_i).$$
    Now we seek to ensure that $\dot{V}(t)\leq -\epsilon$ for some $\epsilon>0$ whenever $\|\tilde{e}\|_\infty\geq r,$ for some $0<r<1$. To that end, note that the above can be conservatively upper bounded as
    $$\dot{V}\leq -w_\mathrm{min}r(r-d_\mathrm{max})+ Kw_\mathrm{max}d^2_\mathrm{max},$$
    where $w_\mathrm{min}=\min_i w_i$ and $w_\mathrm{max}=\max_i w_i$. A $d_\mathrm{max}$ that ensures that this upper bound is smaller than $-\epsilon$ can be found as long as $w_\mathrm{min}r^2 -\epsilon>0$. Hence, all agents in the SCC enter the region $\|\tilde e\|_\infty<r$ in finite time and remain there for all future time.

    Finally, since the remaining agents are influenced by at least one agent in the SCC, the same inductive argument from the leader-follower case (applied to the directed influence paths originating from the SCC) shows that all agents eventually enter and remain in the linear regime, completing the proof.
\end{proof}

\subsection{Time-delayed consensus protocols}\label{app:counterexample}
Here we will illustrate the consequence of the delayed consensus protocol $\Lc_1(z,t)=Dz-\mc{W}(z,t)$ not satisfying Assumption~\ref{ass:kisinvariant}. For simplicity, consider a two-agent system, where one is a leader and both have a constant and equal input $w_0$. The dynamics are then
\begin{align*}
    \dot{z}_0&=w_0\\
    \dot{z}_1&=-z_1(t)+z_0(t-\tau(t))+w_0
\end{align*}
This system can be explicitly solved for $z_0$ and has the solution $z_0(t)=at+z_0(0)$. For the second, consider the case where $\tau(t)=t$ for $t\leq \tau_\mathrm{max}$. Then the solution for $t\leq \tau_\mathrm{max}$ is $z_1(t)=e^{-t}z_1(0)+(a+z_0(0))(1-e^{-t})$. Now, provided that the system is initiated at consensus,  that is $z_1(0)=z_0(0)$, then we see that the agents drift away from each other as $z_0(t)-z_1(t)=a(t-1+e^{-t})$. This shows that the consensus is not an equilibrium solution of this system. Therefore, we cannot expect the agents to reach a consensus when using a delayed consensus protocol for anything other than $\Lc_n$ in the compositional consensus~\eqref{eq:statespace}.

\bibliographystyle{IEEETran}
\bibliography{references}
\vspace{1mm}

\end{document}